\documentclass{article}
\usepackage{graphicx} 

\usepackage{amsmath}
\usepackage{amssymb}
\usepackage{amsthm}
\usepackage{bussproofs}
\usepackage[]{MnSymbol} 

\newtheorem{proposition}{Proposition}
\newtheorem{lemma}{Lemma}
\newtheorem{theorem}{Theorem}

\title{A logic of judgmental existence and its relation to proof irrelevance}
\author{Ivo Pezlar\footnote{Czech Academy of Sciences, Institute of Philosophy, Jilska 1, Prague 110 00.}}
\date{} 

\usepackage[round]{natbib}
\usepackage{hyperref}

\usepackage{footmisc}

\hypersetup{
     colorlinks   = true, 
     citecolor    = blue, 
     linkcolor    = red, 
}

\makeatletter
\DeclareRobustCommand
  \myvdots{\vbox{\baselineskip4\p@ \lineskiplimit\z@
    \hbox{.}\hbox{.}\hbox{.}}}
\makeatother

\newcommand{\mirroredE}{%
  \mathrel{%
      \ensuremath{\equiv\kern-0.88ex\raisebox{0.1ex}{\rule{0.4pt}{0.93ex}}}%
  }%
}

\newcommand{\EXI}{\triangle} 

\newcommand{\TRU}{\Box} 

\newcommand{\existsJ}{%
  \mathrel{%
      \therefore
  }%
}

\newcommand{\starJ}{%
  \mathrel{%
      \star  %
  }%
}

\newcommand{\implyT}{%
  \mathrel{%
      \twoheadrightarrow
  }%
}

\newcommand{\Sopen}{\mathopen{\lsem}}
\newcommand{\Sclose}{\mathclose{\rsem}}

\begin{document}

\maketitle

\paragraph{Abstract.} 
We introduce a simple natural deduction system for reasoning with judgments of the form ``there exists a proof of $\varphi$'' to explore the notion of judgmental existence following Martin-L\"{o}f's methodology of distinguishing between judgments and propositions. In this system, the existential judgment can be internalized into a modal notion of propositional existence that is closely related to truncation modality, a key tool for obtaining proof irrelevance, and lax modality. We provide a computational interpretation in the style of the Curry-Howard isomorphism for the existence modality and show that the corresponding system has some desirable properties such as strong normalization or subject reduction.

\paragraph{Keywords.} judgmental existence, constructive existence, proof irrelevance, truncation, lax modality, Curry-Howard correspondence, homotopy type theory

\paragraph{Funding.} This paper is an outcome of the project Logical Structure of Information Channels, no. 21-23610M, supported by the Czech Science Foundation and realized at the Institute of Philosophy of the Czech Academy of Sciences.

\section{Introduction}
\label{sec:intro}

In the constructive and intuitionist tradition, we can identify two closely related but distinct notions of a proposition $\varphi$ being true, let us call them \textit{true}$_1$ and \textit{true}$_2$, informally:

\begin{enumerate}
    \item[(1)] $\varphi$ is $\textit{true}_1$ $\; \Leftrightarrow \;$ we have a proof $a$ of $\varphi$\footnote{Other common variants: ``we have found a proof $a$ of $\varphi$'', ``we have constructed $a$ of $\varphi$'', ``we know a proof $a$ of $\varphi$'' etc. See, e.g., \cite{nordstrom1990}, p. 29 (``[t]o know that the proposition $A$ is true is to have an element $a$ in $A$''), \cite{troelstra1988}, p. 4 (``[a] statement is \textit{true} if we have proof of it''), \cite{hottbook2013}, p. 41 (``to show that a proposition is true [\ldots] corresponds to exhibiting an element of the type corresponding to that proposition'').}
    
    \item[(2)] $\varphi$ is $\textit{true}_2$ $\; \Leftrightarrow \;$ there exists a proof of $\varphi$\footnote{Other common variants: ``$\varphi$ has a proof'', ``$\varphi$ has been proved''. See, e.g., \cite{dummett1975}, p. 31 (``a mathematical statement is intuitionistically true if there exists an (intuitionistic) proof of it''), \cite{martinlof1991}, p. 141 (``[i]ntuitionistically, truth of a proposition is analyzed as existence of a proof; a proposition is true if there exists a proof of it''), \cite{coquand1989}, p. 4 (``we say that $\varphi$ is true [\ldots] iff $\varphi$ has a proof''). Notice also the subtle difference between ``we have a proof of $\varphi$'' vs. ``$\varphi$ has a proof'' reflecting the difference between (1) and (2).}
\end{enumerate}

\noindent At first glance, these notions might seem interchangeable but there is a subtle and important difference between them. To call something true according to (1) we need to have the actual proof $a$, i.e., to know its construction. However, to call something true according to (2), we just need to know that there exists some proof but what exactly its construction looks like is not important. Also, note that there is a clear difference in order of conceptual priority between (1) and (2): first, we have to have a proof $a$ of $\varphi$ to be able to claim that it exists, i.e., \textit{true}$_2$ is a secondary notion to \textit{true}$_1$.

The notion of being true in the sense of (1) is tightly associated with the idea of proof-relevant approaches to logic and mathematics (which came into dominance especially with the discovery of the Curry-Howard correspondence between propositions and types, \cite{curry1952b}, \cite{howard1980}), while the notion of being true in the sense of (2) is connected with proof-irrelevant approaches:\footnote{The idea of proof irrelevance goes back at least to \cite{debruijn1974}, reprinted in \cite{nederpelt1994}.}

\begin{itemize}
    \item \textit{proof-relevant approaches}: proofs are treated as proper mathematical objects: it is not enough to know that a proposition is true, we also need to have its proof, i.e., to know its structure. In practice, this means that we can have different proofs for the same proposition.

    \item \textit{proof-irrelevant approaches}: proofs are not treated as proper mathematical objects: it is sufficient to know that a proposition is true, the proofs themselves do not matter beyond the fact that they exist, i.e., their structure is unimportant. In practice, this means any two proofs of the same proposition are considered equal (= the proof irrelevance principle).\footnote{In this paper, we will focus on propositional proof irrelevance, not a judgmental one. In other words, any two proofs of the same will be considered extensionally equal, not intensionally.}    
\end{itemize}

\noindent From this perspective, we can view \textit{true}$_1$ and \textit{true}$_2$ as capturing proof-relevant and proof-irrelevant notions of truth, respectively: \textit{true}$_1$ cares about the structure of proofs, \textit{true}$_2$ does not. This also means that $\varphi$ is $\textit{true}_1$ entails $\varphi$ is  $\textit{true}_2$ (we can always ignore structure if there is some assumed) but not vice versa (if there was no structure assumed, we cannot summon it out of nowhere). So, let us call a proposition $\varphi$ that is $\textit{true}_1$ as \textit{proof-relevantly true} or simply \textit{true} and a proposition $\varphi$ that is $\textit{true}_2$ as \textit{proof-irrelevantly true} or simply \textit{just true}:\footnote{Distinguishing between proof relevance and proof irrelevance with no explicit proof terms might seem odd at first but it is still applicable. For example, $\varphi \textit{ true}$ can be informally interpreted as ``we know there is a proof for $\varphi $ and we know how it looks like'' (i.e., the judgment $\varphi \textit{ true}$ is effectively an abbreviation or coding of $a : \varphi$) while $\varphi \textit{ just true}$ says ``we know there is a proof for $\varphi $ but we don't know how it looks like'' (i.e., we cannot ``decode'' the judgment $\varphi \textit{ just true}$ back to $a : \varphi$).}

\begin{enumerate}
    \item[] $\varphi$ is $\textit{true}_1$ $\; \Leftrightarrow \;$ $\varphi$ is $\textit{proof-relevantly true}$   $\; \Leftrightarrow \;$   $\varphi \textit{ true}$
    
    \item[] $\varphi$ is $\textit{true}_2$ $\; \Leftrightarrow \;$ $\varphi$ is $\textit{proof-irrelevantly true}$   $\; \Leftrightarrow \;$   $\varphi \textit{ just true}$
\end{enumerate}

Now, in a type-theoretic setting, the right-hand side of $\; \Leftrightarrow \;$ in (1), i.e., the judgment ``we have a proof $a$ of $\varphi$'' is formalized as the judgment $a : \varphi$:

\begin{itemize}
    \item[(1$'$)] we have a proof $a$ of $\varphi$ $\; \Leftrightarrow \;$ $a : \varphi$
\end{itemize}

\noindent It is the basic judgment of modern type theories utilizing the Curry-Howard correspondence (e.g., Martin-L\"{o}f's constructive type theory (CTT) \cite{martin-lof1984}, Coquand and Huet's calculus of constructions (CoC) \cite{coquand1988}, Luo's unified theory of dependent types (UTT) \cite{luo1994}, homotopy type theory (HoTT), \cite{hottbook2013}).

Surprisingly, however, the right-hand side of $\; \Leftrightarrow \;$ in (2), i.e., the judgment ``there exists a proof of $\varphi$'' is effectively absent. To our knowledge, there are no type theories explicitly working with basic judgments corresponding to ``there exists a proof of $\varphi$'' (in Section \ref{sec:related} we discuss some related systems, including Martin-L\"{o}f's constructive type theory \cite{martin-lof1984} which includes the judgment $\varphi \textit{ true}$ similar in its informal meaning to our $\varphi \textit{ just true}$). In this paper, we will fill this gap and develop and explore a system that directly deals with a new formal judgment corresponding to the right-hand side of $\; \Leftrightarrow \;$ in (2), i.e., to the informal judgment ``there exists a proof of $\varphi$''. We will denote it as $a \existsJ \varphi$ and call it \emph{judgmental existence} or existential judgment.

Thus, we obtain:

\begin{itemize}
    \item[(2$'$)] there exists a proof of $\varphi$ $\; \Leftrightarrow \;$ $a \existsJ \varphi$
\end{itemize}

\noindent However, first, we consider a purely logical variant of the system, i.e., a non-computational variant without explicit proof expressions where $a \existsJ \varphi$ will be replaced with the judgment of the form $\varphi \textit{ just true}$ with the same meaning ``there exists a proof of $\varphi$''.

Note that from (2) and (2$'$) we also obtain:

\begin{itemize}
    \item[(2$''$)] $\varphi$ is $\textit{true}_2$ $\; \Leftrightarrow \;$ $a \existsJ \varphi$
\end{itemize}

\noindent however, we will prefer the reading of $a \existsJ \varphi$ as ``there exists a proof of $\varphi$'' to ``$\varphi$ is just true'' even though they are interchangeable.

\medskip

The main contribution of this paper is that we devise a simple calculus for reasoning with existential judgments of general form ``there exists a proof of $\varphi$''. More specifically, we introduce two variants of the calculus: a purely logical one for reasoning with existential judgments of the form $\varphi \textit{ just true}$ and a computational one for reasoning with existential judgments of the form $a \existsJ \varphi$. In its present form, the calculus deals only with a fragment of propositional logic containing two operators, implication ($\to$) and propositional existence modality ($\EXI$) used for internalization of the existential judgment.\footnote{The term ``propositional existence'' for propositions of the form $\EXI \varphi$ may not be ideal, as it could be mistaken for ``existential proposition'' which refers to propositions of the form $\exists x : \varphi , \ldots$. However, we prefer this terminology as it aligns well with the term ``judgmental existence''.} We show that the resulting calculus is strongly normalizing and has the subject reduction property. It also turns out that the existence modality operator $\EXI$ is closely related to the notion of propositional truncation $\TRU$, a crucial tool for capturing proof irrelevance in otherwise proof-relevant systems. We do not consider dependent types and proof irrelevance is assumed to be propositional (not judgmental). This is sufficient to illustrate all the key concepts presented in the paper.

Earlier, we have said that there are no type theories with a basic judgment corresponding to ``there exists a proof of $\varphi$''. That is, of course, not to say that the importance of disregarding proofs, i.e., the idea of proof irrelevance, has gone unnoticed. On the contrary, it is well well-recognized and studied issue in the literature.  However, many of the approaches for obtaining proof irrelevance do not explicitly introduce judgmental existence, i.e., judgments of the form $\varphi \textit{ just true}$/$a \existsJ \varphi$, and treat it only indirectly via a propositional modality (most notably the truncation modality\footnote{Also known as squash type, bracket type or subsingleton type. We move away from the usual notation for truncation ``$\Vert \varphi \Vert$'' towards the simpler notation ``$\TRU\varphi$''.} where, e.g., $\TRU\varphi$ can be understood as ``$\varphi$ is inhabited'', see, e.g., \cite{awodey2004}, \cite{gilbert2019}, \cite{hottbook2013}). And even those approaches that take into account the judgmental existence (see, e.g., \cite{valentini1998}, \cite{pfenning2001b}, \cite{reed2002}, \cite{abel2012}), treat the corresponding judgments $\varphi \textit{ just true}$/$a \existsJ \varphi$ only as a notational abbreviation for other judgments (we discuss the related approaches more in Section \ref{sec:related}). To our knowledge, there are currently no approaches that would treat $\varphi \textit{ just true}$/$a \existsJ \varphi$ as a basic judgment.

Finally, it is important to note that the judgment $\varphi \textit{ just true}$ expresses a notion of existence that is distinct from the notion of existence typically associated with existential propositions, i.e., propositions formed via the existential quantifier $\exists$. To see this, consider the following observation by \cite{martinlof1993} (originally made in the context of a related judgment $\varphi \textit{ exists}$, see Section \ref{sec:related}). If we were to attempt to express the notion of existence behind $\varphi \textit{ just true}$ via $\exists$ we would be inadvertently forming a new proposition of the form $\exists x : \varphi \ldots$ and since it would be a proposition, we should also be able to explain what it means for it to be true in the sense of (2). This would get us ``a proposition $\exists x : \varphi \ldots$ is $\textit{true}_2$ $\Leftrightarrow$ there exists a proof of it'' but here the notion of judgmental existence appears once more. And, of course, any further attempts to explain it via $\exists$ again would lead to an infinite regress. Thus, propositions of the form $\exists x : \varphi \ldots$ (as do all other propositions) presuppose this notion of judgmental existence.

\medskip

\noindent \textit{Structure}. In Section \ref{sec:logical}, we introduce a purely logical variant of the calculus for reasoning with existential judgments of the form $\varphi \textit{ just true}$. In Section \ref{sec:computational}, we introduce its computational variant in the style of the Curry-Howard correspondence for reasoning with existential judgments of the form $a \existsJ \varphi$ with explicit proof expressions. Furthermore, we show that the system has some desirable properties such as strong normalization and subject reduction. In Section \ref{sec:truncation}, we explore the connection between existential judgment and propositional truncation modality. In Section \ref{sec:related}, we survey some related work, including the relation of existence modality to lax modality.

\section{A logical variant}
\label{sec:logical}

In this section, we introduce a natural deduction calculus for reasoning with judgments of the form $\varphi \textit{ true}$ (= $\textit{true}_1$) and $\varphi \textit{ just true}$ (= $\textit{true}_2$). The former tells us that $\varphi$ is proof-relevantly true (i.e., the structure of the corresponding proof of $\varphi$ is known/relevant to us), while the latter tells us that $\varphi$ is proof-irrelevantly just true (i.e., the structure of the proof of $\varphi$ is unknown/irrelevant to us). In other words, it only tells us that there exists a proof of $\varphi$ and nothing more. We consider a fragment of propositional logic containing only implication $\to$ and existence modality $\EXI$.

We show that the system satisfies the structural properties of exchange, weakening, and contraction for both judgments $\varphi \textit{ true}$ and $\varphi \textit{ just true}$ and we prove the substitution lemma for all possible combinations of these judgments. Furthermore, we prove some basic propositions about the relationship of $\to$ and $\EXI$.

\subsection{Formal system}

\noindent Language:

\begin{align*}
\text{Propositions } \varphi, \psi & \;::=\; p \mid \varphi \to \psi \mid \EXI\varphi \\
\text{Hypotheses } \Gamma & \;::=\; \cdot \mid \Gamma, \varphi \textit{ true} \\
\end{align*}

\noindent Categorical judgments:
$$\varphi \textit{ true} \qquad \varphi \textit{ just true} $$

\noindent Hypothetical judgments:\footnote{Note that existential judgments can appear only in the consequent of the hypothetical judgments. In other words, we do not allow assumptions of the form $\varphi \textit{ just true}$.}

$$\Gamma \vdash \varphi \textit{ true} \qquad \Gamma \vdash \varphi \textit{ just true}$$

\noindent These judgments are governed by the following rules and rules schemata:

\begin{center}
\AxiomC{}
\RightLabel{\footnotesize \footnotesize \textsc{\textsc{hyp}}}
\UnaryInfC{$\Gamma, \varphi \textit{ true} \vdash \varphi \textit{ true}  $}
\DisplayProof
\quad
\AxiomC{$\Gamma \vdash \varphi \textit{ true}$}
\RightLabel{\footnotesize \footnotesize \textsc{\textsc{just}}}
\UnaryInfC{$\Gamma \vdash \varphi \textit{ just true}$}
\DisplayProof
\end{center}

\begin{center}
\AxiomC{$\Gamma, \varphi \textit{ true} \vdash \psi \textit{ star}$}
\RightLabel{\footnotesize \footnotesize $\to$I$\starJ$}
\UnaryInfC{$\Gamma \vdash \varphi \to \psi \textit{ star}  $}
\DisplayProof
\quad
\AxiomC{$\Gamma \vdash \varphi \to \psi \textit{ star}  $}
\AxiomC{$\Gamma \vdash \varphi \textit{ true}  $}
\RightLabel{\footnotesize $\to$E$\starJ$}
\BinaryInfC{$\Gamma \vdash \psi \textit{ star}  $}
\DisplayProof
\end{center}

\begin{center}
\AxiomC{$\Gamma \vdash \varphi \textit{ just true}$}
\RightLabel{\footnotesize $\EXI$I$\starJ$}
\UnaryInfC{$\Gamma \vdash \EXI\varphi \textit{ star}  $}
\DisplayProof
\quad
\AxiomC{$\Gamma \vdash \EXI\varphi \textit{ star}$}
\AxiomC{$\Gamma , \varphi \textit{ true} \vdash \gamma \textit{ just true}$}
\RightLabel{\footnotesize $\EXI$E$\starJ$}
\BinaryInfC{$\Gamma \vdash \gamma \textit{ just true}  $}
\DisplayProof
\end{center}

\noindent where $\textit{star}$ stands for either ``$\textit{true}$'' or ``$\textit{just true}$'' (all occurencens must be consistent). Thus, e.g., $\to$I and $\to$E will indicate instances of the rules $\to$I$\starJ$ and $\to$E$\starJ$ with $\textit{true}$, while $\to$Ij and $\to$Ej with $\textit{just true}$. Analogously, $\EXI$I and $\EXI$E will indicate instances of  $\EXI$I$\starJ$ and $\EXI$E$\starJ$ with $\textit{true}$, while $\EXI$Ij and $\EXI$Ej with $\textit{just true}$.

Finally, depending on the occurrences of existential judgments, we will have the following two substitution rules:

\begin{center}
\AxiomC{$\Gamma \vdash \varphi \textit{ star}$}
\AxiomC{$\Gamma, \varphi \textit{ true} \vdash \gamma \textit{ star}$}
\RightLabel{\footnotesize \textsc{sub1}}
\BinaryInfC{$\Gamma \vdash \gamma \textit{ star}$}
\DisplayProof

\medskip
\medskip

\AxiomC{$\Gamma \vdash \varphi \textit{ true}$}
\AxiomC{$\Gamma, \varphi \textit{ true} \vdash \gamma \textit{ just true}$}
\RightLabel{\footnotesize \textsc{sub2}}
\BinaryInfC{$\Gamma \vdash \gamma \textit{ just true}$}
\DisplayProof
\end{center}

\medskip 

\noindent \textit{Comments on} \textsc{\textsc{just}} \textit{rule}. This rule specifies the meaning of the judgment $\varphi \textit{ just true}$ (together with the appropriate instance of the \textsc{sub1} rule, see below). As discussed in Section \ref{sec:intro}, it captures the idea that if we have the actual proof of $\varphi$, i.e., if we know that $\varphi$ is (proof-relevantly) true, then we can derive that there exists a proof of $\varphi$, i.e., that $\varphi$ is (proof-irrelevantly) true. In other words, once we know the structure of a proof, we can always choose to disregard it.

\medskip

\noindent \textit{Comments on} $\to$I$\starJ$ and $\to$E$\starJ$. Note that the true variants $\to$I and $\to$E, where $\textit{star}$ is replaced with $\textit{true}$, are just ordinary implication introduction and elimination rules for defining the meaning of the connective $\to$ as known from natural deduction (\cite{prawitz1965}). The just true variants $\to$Ij and $\to$Ej, where $\textit{star}$ is replaced with $\textit{just true}$, are more interesting. $\to$Ij states that if we can derive that $\psi$ is just true under the assumption that $\varphi$ is true, then we can derive that $\varphi \to \psi$ is just true. So, it is a proof-irrelevant variant of the standard implication introduction rule. Analogously, $\to$Ej is a proof-irrelevant variant of the standard implication elimination rule: it states that if $\varphi \to \psi$ is just true and if $\varphi$ is true, then we can derive that $\psi$ is just true. Later, we will show that $\to$Ej can be actually considered as a derived rule, however, this is not the case for $\to$Ij, unless we adopt an additional principle (Proposition \ref{prop:lax_derivability} in Section \ref{sec:related}). Both true and just true variants of $\to$E$\starJ$ are locally sound and complete (see Section \ref{sec:reductions}). 

\medskip

\noindent \textit{Comments on} $\EXI$I$\starJ$ and $\EXI$E$\starJ$. Let us start with the true variants $\EXI$I and $\EXI$E. The introduction rule $\EXI$I internalizes the judgment $\varphi \textit{ just true}$ by turning it into a proposition, specifically into the true proposition $\EXI\varphi$. In other words, it defines the meaning of the proof existence modality $\EXI$.\footnote{\cite{klev2023} offers an interesting account of modalities understood as illocutionary forces of speech acts. From this perspective, it would be interesting to explore to what kind of illocutionary force the modality $\EXI$ might intuitively correspond. A natural candidate seems to be a commissive force or some variant of it: $\EXI \varphi$ can be regarded as a promise that we have a proof of $\varphi$ but we are not willing to show it.} The elimination rule $\EXI$E then tells us how can we ``unpack'' the definition of $\EXI$ and return back to the existential judgment level with $\gamma \textit{ just true}$. Now, let us consider the just true variants $\EXI$Ij and $\EXI$Ej. $\EXI$Ij states that if we can derive that $\varphi$ is just true, then we can also derive that $\EXI \varphi$ is just true. So, it is a proof-irrelevant variant of $\EXI$I. Analogously, $\EXI$Ej is a proof-irrelevant variant of $\EXI$E: it states that if $\EXI \varphi$ is just true and if $\gamma$ is just true under the assumption that $\varphi$ is true, then we can derive that $\gamma$ is just true as well. Later, we will show that both $\EXI$Ij and $\EXI$Ej can be actually considered as derived rules (Proposition \ref{prop:lax_derivability} in Section \ref{sec:related}). Both true and just true variants of $\EXI$E$\starJ$ are locally sound and complete (see Section \ref{sec:reductions}). 

\medskip

Now, to simplify some upcoming proofs, let us introduce a new connective $\varphi \implyT \psi$ called existential implication and defined as $\varphi \to \EXI\psi$. Alternatively, we could also specify it directly via the following I/E rules:

\begin{center}
\AxiomC{$\Gamma, \varphi \textit{ true} \vdash \psi \textit{ just true}$}
\RightLabel{\footnotesize $\implyT$I$\starJ$}
\UnaryInfC{$\Gamma \vdash \varphi \implyT \psi \textit{ star}  $}
\DisplayProof
\quad
\AxiomC{$\Gamma \vdash \varphi \implyT \psi \textit{ star}  $}
\AxiomC{$\Gamma \vdash \varphi \textit{ true}  $}
\RightLabel{\footnotesize $\implyT$E$\starJ$}
\BinaryInfC{$\Gamma \vdash \psi \textit{ just true}  $}
\DisplayProof
\end{center}

\medskip

\noindent Note that the rule $\implyT$I$\starJ$ essentially internalizes the hypothetical judgment $\Gamma, \varphi \textit{ true} \vdash \psi \textit{ just true}$ as a proposition $\varphi \implyT \psi$ that can be either $\textit{true}$ or $\textit{just true}$.

Also note that these I/E rules become derivable when we consider $\varphi \implyT \psi$ to be defined as $\varphi \to \EXI\psi$:

\begin{center}
    \AxiomC{$\Gamma , \varphi \textit{ true} \vdash \psi \textit{ just true}$}
    \RightLabel{\footnotesize $\EXI$I$\starJ$}
    \UnaryInfC{$\Gamma , \varphi \textit{ true} \vdash \EXI \psi \textit{ star}$}
    \RightLabel{\footnotesize $\to$I$\starJ$}
    \UnaryInfC{$\Gamma \vdash  \varphi \to \EXI \psi \textit{ star}$}
    \DisplayProof \quad
    
    \medskip
    \medskip
    
    \AxiomC{$\Gamma \vdash \varphi \to \EXI \psi \textit{ star}  $}
    \AxiomC{$\Gamma \vdash \varphi \textit{ true}  $}
    \RightLabel{\footnotesize $\to$E$\starJ$}
    \BinaryInfC{$\Gamma \vdash \EXI \psi \textit{ star}  $}
    \AxiomC{}
    \RightLabel{\footnotesize \textsc{hyp}}
    \UnaryInfC{$\Gamma , \psi \textit{ true } \vdash \psi \textit{ true} $}
    \RightLabel{\footnotesize \textsc{just}}
    \UnaryInfC{$\Gamma , \psi \textit{ true } \vdash \psi \textit{ just true} $}
    \RightLabel{\footnotesize $\EXI$E$\starJ$}
    \BinaryInfC{$\Gamma \vdash \psi \textit{ just true}  $}
    \DisplayProof
\end{center}

\subsection{Basic properties}

\begin{lemma} (Exchange) If $\Gamma, \varphi \textit{ true}, \psi \textit{ true}  \vdash \gamma \textit{ star}$, then $\Gamma, \psi \textit{ true} , \varphi \textit{ true} \vdash \gamma \textit{ star}$.
\end{lemma}

\begin{proof}
By structural induction on the derivations of $\Gamma, \varphi \textit{ true}, \psi \textit{ true}  \vdash \gamma \textit{ true}$ and $\Gamma, \varphi \textit{ true}, \psi \textit{ true}  \vdash \gamma \textit{ just true}$, respectively.
\end{proof}

\noindent For example, let us consider $\Gamma, \varphi \textit{ true}, \psi \textit{ true}  \vdash \gamma \textit{ just true}$ for the case of $\EXI$E:

\begin{center}
$\mathcal{D}$ = $\AxiomC{$\mathcal{D}'$}\noLine\UnaryInfC{$\Gamma , \tau \textit{ true} , \theta \textit{ true} \vdash \EXI\psi \textit{ true}$}
\AxiomC{$\mathcal{D}''$}\noLine\UnaryInfC{$\Gamma , \tau \textit{ true} , \theta \textit{ true} , \psi \textit{ true} \vdash \sigma \textit{ just true}$}
\RightLabel{\footnotesize $\EXI$E}
\BinaryInfC{$\Gamma, \tau \textit{ true} , \theta \textit{ true} \vdash \sigma \textit{ just true}  $}\DisplayProof$
\end{center}

\begin{enumerate}
    \item $\Gamma , \theta \textit{ true} , \tau \textit{ true} \vdash \EXI\psi \textit{ true}$ (by exchange for judgments of the form $\gamma \textit{ true}$ on $\mathcal{D}'$)
    
    \item $\Gamma , \theta \textit{ true} , \tau \textit{ true} , \psi \textit{ true} \vdash \sigma \textit{ just true}$ (by induction hypothesis, IH, on $\mathcal{D}''$)

    \item $\Gamma, \theta \textit{ true} , \tau \textit{ true} \vdash \sigma \textit{ just true}  $ (by $\EXI$E on 1. and 2.)
\end{enumerate}

\begin{lemma} (Weakening)  If $\Gamma \vdash \gamma \textit{ star}$, then $\Gamma, \varphi \textit{ true} \vdash \gamma \textit{ star} $.
\end{lemma}

\noindent \textit{Note}. Since we disallow just true assumptions, we need to consider only two cases.

\begin{proof}
By structural induction on the derivations of $\Gamma \vdash \gamma \textit{ true}$ and $\Gamma \vdash \gamma \textit{ just true}$, respectively.
\end{proof}

\noindent For example, let us consider $\Gamma \vdash \gamma \textit{ just true}$ for the case of $\EXI$E:

\begin{center}
$\mathcal{D}$ = $\AxiomC{$\mathcal{D}'$}\noLine\UnaryInfC{$\Gamma \vdash \EXI\psi \textit{ true}$}
\AxiomC{$\mathcal{D}''$}\noLine\UnaryInfC{$\Gamma , \psi \textit{ true} \vdash \sigma \textit{ just true}$}
\RightLabel{\footnotesize $\EXI$E}
\BinaryInfC{$\Gamma \vdash \sigma \textit{ just true}  $}\DisplayProof$
\end{center}

\begin{enumerate}
    \item $\Gamma, \varphi \textit{ true} \vdash \EXI\psi \textit{ true}$ (by weakening for judgments of the form $\gamma \textit{ true}$ on $\mathcal{D}'$)
    
    \item $\Gamma, \psi \textit{ true}, \varphi \textit{ true} \vdash \sigma \textit{ just true}$ (by IH on $\mathcal{D}''$)
    
    \item $\Gamma , \varphi \textit{ true} \vdash \sigma \textit{ just true}$ (by exchange on 2. and $\EXI$E on 1. and 2.)
\end{enumerate}

\begin{lemma} (Contraction) If $\Gamma , \varphi \textit{ true} , \varphi \textit{ true} \vdash \gamma \textit{ star}$, then $\Gamma, \varphi \textit{ true} \vdash \gamma \textit{ star} $.
\end{lemma}

\begin{proof}
By structural induction.
\end{proof}

\begin{lemma} (Substitution lemma)
\begin{enumerate}

    \item[](\textsc{sub1}) If $\Gamma \vdash \varphi \textit{ star}$ and $\Gamma, \varphi \textit{ true} \vdash \gamma  \textit{ star}$ then $\Gamma \vdash \gamma \textit{ star}$.

    \item[](\textsc{sub2}) If $\Gamma \vdash \varphi \textit{ true}$ and $\Gamma, \varphi \textit{ true} \vdash \gamma  \textit{ just true}$ then $\Gamma \vdash \gamma \textit{ just true}$.

\end{enumerate}
\end{lemma}

\begin{proof}
(\textsc{sub1}, \textit{star} = \textit{true}) By structural induction on the derivation of $\Gamma , \varphi \textit{ true} \vdash \gamma \textit{ true}$.   
\end{proof}

\begin{proof} (\textsc{sub1}, \textit{star} = \textit{just true}) By structural induction on the derivation of $\Gamma \vdash \varphi \textit{ just true}$.
\end{proof}

\noindent For example, let us consider again the case of $\EXI$E:

\begin{center}
$\mathcal{D}$ = $\AxiomC{$\mathcal{D}'$}\noLine\UnaryInfC{$\Gamma \vdash \EXI\psi \textit{ true}$}
\AxiomC{$\mathcal{D}''$}\noLine\UnaryInfC{$\Gamma , \psi \textit{ true} \vdash \varphi \textit{ just true}$}
\RightLabel{\footnotesize $\EXI$E}
\BinaryInfC{$\Gamma \vdash \varphi \textit{ just true}  $}\DisplayProof$
\end{center}

\begin{enumerate}

    \item $\Gamma , \varphi \textit{ true} \vdash \gamma \textit{ just true}$ (by assumption)

    \item $\Gamma \vdash \psi \textit{ just true}$ (by inversion\footnote{Inversion is a metalevel reasoning principle for the level of judgments that essentially states: if the conclusion of a rule holds, so must its premise(s).} on $\mathcal{D}'$)

    \item $\Gamma \vdash \varphi \textit{ just true}$ (by IH on 2. and $\mathcal{D}''$)

    \item $\Gamma \vdash \EXI \varphi \textit{ true}$ (by $\EXI$I on 3.) 

    \item $\Gamma \vdash \gamma \textit{ just true}$ (by $\EXI$E on 1. and 4.) 

\end{enumerate}

\begin{proof}
(\textsc{sub2}) By structural induction on the derivation of $\Gamma, \varphi \textit{ true} \vdash \gamma \textit{ just true}$.   
\end{proof}

\noindent For example, let us consider the case of $\EXI$E:

\begin{center}
$\mathcal{D}$ = $\AxiomC{$\mathcal{D}'$}\noLine\UnaryInfC{$\Gamma , \varphi \textit{ true} \vdash \EXI\psi \textit{ true}$}
\AxiomC{$\mathcal{D}''$}\noLine\UnaryInfC{$\Gamma , \varphi \textit{ true} , \psi \textit{ true} \vdash \sigma \textit{ just true}$}
\RightLabel{\footnotesize $\EXI$E}
\BinaryInfC{$\Gamma , \varphi \textit{ true} \vdash \sigma \textit{ just true}  $}\DisplayProof$
\end{center}

\begin{enumerate}
        
    \item $\Gamma \vdash \varphi \textit{ true}$ (by assumption)

    \item $\Gamma \vdash \EXI\psi \textit{ true} $ (by \textsc{sub1} on 1. and $\mathcal{D}'$)

    \item $\Gamma , \psi \textit{ true} \vdash \varphi \textit{ true}$ (by weakening on $\mathcal{D}'$)

    \item $\Gamma , \psi \textit{ true}, \varphi \textit{ true}  \vdash \sigma \textit{ just true}$ (by exchange on $\mathcal{D}''$)

    \item $\Gamma , \psi \textit{ true} \vdash \sigma \textit{ just true}$ (by IH on 3. and 4.)

    \item $\Gamma \vdash \sigma \textit{ just true}$ (by $\EXI$E on 2. and 5.)
   
\end{enumerate}

\subsection{Reduction and expansion rules}
\label{sec:reductions}

\noindent Reduction and expansion for $\EXI\varphi \textit{ star}$:

\begin{center}
\AxiomC{$\mathcal{D}$}
\noLine
\UnaryInfC{$\Gamma \vdash \varphi \textit{ just true}$}
\RightLabel{\footnotesize $\EXI$I$\starJ$}
\UnaryInfC{$\Gamma \vdash \EXI\varphi \textit{ star}$}
\AxiomC{$\mathcal{D'}$}
\noLine
\UnaryInfC{$\Gamma ,  \varphi  \textit{ true} \vdash \gamma \textit{ just true}$}
\RightLabel{\footnotesize $\EXI$E$\starJ$}
\BinaryInfC{$\Gamma \vdash \gamma \textit{ just true}$}
\DisplayProof
$\Rightarrow_R$
\AxiomC{$\mathcal{D''}$}
\noLine
\UnaryInfC{$\Gamma \vdash \gamma \textit{ just true}$}
\DisplayProof
\end{center}

\noindent where $\mathcal{D}''$ is constructed from $\mathcal{D'}$ by replacing instances of the hypothesis $\varphi \textit{ true}$ with $\mathcal{D}$ (via \textsc{sub1}).\footnote{Note that these reductions do not operate on natural deduction derivations with formulas but on demonstrations with judgments. Thus our natural-deduction-like phrasing ``replacing instances of the hypothesis $\varphi \textit{ true}$ true with $\mathcal{D}$'' is somewhat an abuse of language (this can also be observed in \cite{pfenning2001}). We thank Ansten Klev for this note. A proper judgmental explanation of these reduction metarules will come later in Section \ref{sec:computational} where we introduce reducibility judgments of the form $\Gamma \vdash e \Longrightarrow_R e'$ corresponding to these reductions.}

\medskip

\begin{center}
\AxiomC{$\mathcal{D}$}
\noLine
\UnaryInfC{$\Gamma \vdash \EXI\varphi \textit{ star}$}
\DisplayProof
$\Rightarrow_E$
\AxiomC{$\mathcal{D}$}
\noLine
\UnaryInfC{$\Gamma \vdash \EXI\varphi \textit{ star}$}
\AxiomC{}
\RightLabel{\footnotesize \textsc{hyp}}
\UnaryInfC{$\Gamma, \varphi \textit{ true} \vdash \varphi \textit{ true}$}
\RightLabel{\footnotesize \textsc{just}}
\UnaryInfC{$\Gamma, \varphi \textit{ true} \vdash \varphi \textit{ just true}$}
\RightLabel{\footnotesize $\EXI$E$\starJ$}
\BinaryInfC{$\Gamma \vdash \varphi \textit{ just true}$}
\RightLabel{\footnotesize $\EXI$I$\starJ$}
\UnaryInfC{$\Gamma \vdash \EXI\varphi \textit{ star}$}
\DisplayProof
\end{center}

\noindent Reduction and expansion for $\varphi \to \psi \textit{ star}$:

\begin{center}
\AxiomC{$\mathcal{D}$}
\noLine
\UnaryInfC{$\Gamma , \varphi \textit{ true} \vdash \psi \textit{ star}$}
\RightLabel{\footnotesize $\to$I$\starJ$}
\UnaryInfC{$\Gamma \vdash \varphi \to \psi \textit{ star}$}
\AxiomC{$\mathcal{D'}$}
\noLine
\UnaryInfC{$\Gamma \vdash \varphi \textit{ true}$}
\RightLabel{\footnotesize $\to$E$\starJ$}
\BinaryInfC{$\Gamma \vdash \psi \textit{ star}$}
\DisplayProof
$\Rightarrow_R$
\AxiomC{$\mathcal{D''}$}
\noLine
\UnaryInfC{$\Gamma \vdash \psi \textit{ star}$}
\DisplayProof
\end{center}

\noindent where $\mathcal{D}''$ is constructed from $\mathcal{D}$ by replacing instances of the hypothesis $\varphi \textit{ true}$ with $\mathcal{D'}$ (via \textsc{sub1}) if $\textit{star = true}$, otherwise, i.e., if $\textit{star = exist}$, $\mathcal{D}''$ is constructed from $\mathcal{D}$ by replacing instances of the hypothesis $\varphi \textit{ true}$ with $\mathcal{D'}$ (via \textsc{sub2}).

\begin{center}
\small
\AxiomC{$\mathcal{D}$}
\noLine
\UnaryInfC{$\Gamma \vdash \varphi \to \psi \textit{ star}$}
\DisplayProof
$\Rightarrow_E$
\AxiomC{$\mathcal{D}$}
\noLine
\UnaryInfC{$\Gamma \vdash \varphi \to \psi \textit{ star}$}
\RightLabel{\footnotesize \textsc{w}}
\UnaryInfC{$\Gamma, \varphi \textit{ true} \vdash \varphi \to \psi \textit{ star}$}
\AxiomC{}
\RightLabel{\footnotesize \textsc{hyp}}
\UnaryInfC{$\Gamma, \varphi \textit{ true} \vdash \varphi \textit{ true}$}
\RightLabel{\footnotesize $\to$E$\starJ$}
\BinaryInfC{$\Gamma, \varphi \textit{ true} \vdash \psi\textit{ star}$}
\RightLabel{\footnotesize $\to$I$\starJ$}
\UnaryInfC{$\Gamma \vdash \varphi \to \psi \textit{ star}$}
\DisplayProof
\end{center}

\noindent Reduction and expansion for $\varphi \implyT \psi \textit{ star}$\,:

\begin{center}
\AxiomC{$\mathcal{D}$}
\noLine
\UnaryInfC{$\Gamma , \varphi \textit{ true} \vdash \psi \textit{ just true}$}
\RightLabel{\footnotesize $\implyT$I$\starJ$}
\UnaryInfC{$\Gamma \vdash \varphi \implyT \psi \textit{ star}$}
\AxiomC{$\mathcal{D'}$}
\noLine
\UnaryInfC{$\Gamma \vdash \varphi \textit{ true}$}
\RightLabel{\footnotesize $\implyT$E$\starJ$}
\BinaryInfC{$\Gamma \vdash \psi \textit{ just true}$}
\DisplayProof
$\Rightarrow_R$
\AxiomC{$\mathcal{D''}$}
\noLine
\UnaryInfC{$\Gamma \vdash \psi \textit{ just true}$}
\DisplayProof
\end{center}

\noindent where $\mathcal{D}''$ is constructed from $\mathcal{D}$ by replacing instances of the hypothesis $\varphi \textit{ true}$ with $\mathcal{D'}$ (via \textsc{sub2}).

\begin{center}
\small
\AxiomC{$\mathcal{D}$}
\noLine
\UnaryInfC{$\Gamma \vdash \varphi \implyT \psi \textit{ star}$}
\DisplayProof
$\Rightarrow_E$
\AxiomC{$\mathcal{D}$}
\noLine
\UnaryInfC{$\Gamma \vdash \varphi \implyT \psi \textit{ star}$}
\RightLabel{\footnotesize \textsc{w}}
\UnaryInfC{$\Gamma , \varphi \textit{ true} \vdash \varphi \implyT \psi \textit{ star}$}
\AxiomC{}
\RightLabel{\footnotesize \textsc{hyp}}
\UnaryInfC{$\Gamma, \varphi \textit{ true} \vdash \varphi \textit{ true}$}
\RightLabel{\footnotesize $\implyT$E$\starJ$}
\BinaryInfC{$\Gamma, \varphi \textit{ true} \vdash \psi\textit{ just true}$}
\RightLabel{\footnotesize $\implyT$I$\starJ$}
\UnaryInfC{$\Gamma \vdash \varphi \implyT \psi \textit{ star}$}
\DisplayProof
\end{center}

\paragraph{Some propositions.} \phantom{} 

\begin{proposition} \label{prop:1}
$\Gamma \vdash \varphi \to \psi \textit{ just true}$ \,iff\,  $\Gamma \vdash  \varphi \to \EXI \psi \textit{ true}$.    
\end{proposition}

\noindent Informally, it tells us that if $\varphi \to \psi$ is just true, i.e., if there exists a proof of it, then we can derive that $\varphi \to \EXI\psi$ is true, i.e., that we have a proof for it, and vice versa. In other words, a proof of $\varphi \to \EXI\psi$ carries no more information than the existence of a proof of $\varphi \to \psi$.

\begin{proof}
To simplify the proof, we use $\varphi \implyT \psi$ instead of $\varphi \to \EXI \psi$ (recall that $\implyT$ is defined via $\to$ and $\EXI$).

\medskip

\begin{prooftree}
\AxiomC{$\Gamma \vdash \varphi \implyT \psi \textit{ true}$}
\RightLabel{\footnotesize \textsc{w}}
\UnaryInfC{$\Gamma  , \varphi \textit{ true}  \vdash \varphi \implyT \psi \textit{ true}$}
\AxiomC{}
\RightLabel{\footnotesize \textsc{hyp}}
\UnaryInfC{$\Gamma , \varphi \textit{ true} \vdash \varphi \textit{ true}$}
\RightLabel{\footnotesize $\implyT$E}
\BinaryInfC{$\Gamma , \varphi \textit{ true} \vdash \psi \textit{ just true}$}
\RightLabel{\footnotesize $\to$Ij }
\UnaryInfC{$\Gamma \vdash \varphi \to \psi \textit{ just true}$}
\end{prooftree}

\begin{prooftree}
    \AxiomC{$\Gamma \vdash \varphi \to \psi \textit{ just true}$}
    \RightLabel{\footnotesize \textsc{w}}
    \UnaryInfC{$\Gamma , \varphi \textit{ true} \vdash \varphi \to \psi \textit{ just true}$}
    \AxiomC{}
    \RightLabel{\footnotesize \textsc{hyp}}
    \UnaryInfC{$\Gamma , \varphi \textit{ true} \vdash \varphi \textit{ true}$}
    \RightLabel{\footnotesize $\to$Ej }
    \BinaryInfC{$\Gamma , \varphi \textit{ true} \vdash \psi \textit{ just true}$}
    \RightLabel{\footnotesize $\implyT$I}
    \UnaryInfC{$\Gamma \vdash \varphi \implyT \psi \textit{ true}$}
\end{prooftree}

\end{proof}

\medskip

\noindent \textit{Note}. If we were interested only in the notion of proof-relevant $\textit{true}$, we could simply transform Proposition 1 into:

    \begin{enumerate}
    \item[] $\Gamma \vdash \EXI (\varphi \to \psi) \textit{ true}$ \,iff\,  $\Gamma \vdash  \varphi \to \EXI \psi \textit{ true}$.
    \end{enumerate}
    
\noindent via the application of $\EXI$I on the first judgment (analogously for the following propositions). Also, note that we can prove that if $\Gamma \vdash \varphi \to \psi \textit{ true}$ then $\Gamma \vdash  \varphi \to \EXI \psi \textit{ true}$ but not vice versa.

\begin{prooftree}
    \AxiomC{$\Gamma \vdash \varphi \to \psi \textit{ true}$}
    \RightLabel{\footnotesize \textsc{w}}
    \UnaryInfC{$\Gamma , \varphi \textit{ true} \vdash \varphi \to \psi \textit{ true}$}
    \AxiomC{}
    \RightLabel{\footnotesize \textsc{hyp}}
    \UnaryInfC{$\Gamma , \varphi \textit{ true} \vdash \varphi \textit{ true}$}
    \RightLabel{\footnotesize $\to$E}
    \BinaryInfC{$\Gamma , \varphi \textit{ true} \vdash \psi \textit{ true}$}
    \RightLabel{\footnotesize \textsc{just}}
    \UnaryInfC{$\Gamma , \varphi \textit{ true} \vdash \psi \textit{ just true}$}
    \RightLabel{\footnotesize $\EXI$I}
    \UnaryInfC{$\Gamma , \varphi \textit{ true} \vdash \EXI\psi \textit{ true}$}
    \RightLabel{\footnotesize $\to$I}
    \UnaryInfC{$\Gamma \vdash \varphi \to \EXI \psi \textit{ true}$}
\end{prooftree}

\noindent However, from $\Gamma \vdash  \varphi \to \EXI \psi \textit{ true}$ we can derive $\Gamma \vdash   \EXI \varphi \to \EXI \psi \textit{ true}$ and vice versa. 

\begin{proposition} \label{prop:2}
$\Gamma \vdash  \varphi \to \EXI \psi \textit{ true}$ \,iff\,  $\Gamma \vdash \EXI \varphi \to \EXI \psi \textit{ true}$.   
\end{proposition}

\begin{proof}
To shorten this proof, we will write ``$\varphi$'' instead of ``$\varphi \textit{ true}$'' and ``$\varphi \textit{ just}$'' instead of ``$\varphi \textit{ just true}$'' (we will adopt this notation for all upcoming longer proofs):
\end{proof}

\begin{prooftree}
\def\defaultHypSeparation{\hskip 0.05in}
\small
    \AxiomC{}
    \RightLabel{\footnotesize \textsc{hyp}}
    \UnaryInfC{$\Gamma , \EXI \varphi \vdash \EXI \varphi$}

    \AxiomC{$\Gamma \vdash \varphi \to \EXI \psi$}
    \RightLabel{\footnotesize \textsc{w}}
    \UnaryInfC{$\Gamma , \varphi \vdash \varphi \to \EXI \psi$}

    \AxiomC{}
    \RightLabel{\footnotesize \textsc{hyp}}
    \UnaryInfC{$\Gamma , \varphi \vdash \varphi$}

    \RightLabel{\footnotesize $\to$E}
    \BinaryInfC{$\Gamma , \varphi \vdash \EXI \psi$}

    \AxiomC{}
    \RightLabel{\footnotesize \textsc{hyp}}
    \UnaryInfC{$\Gamma , \psi \vdash \psi$}
    \RightLabel{\footnotesize \textsc{w}}
    \UnaryInfC{$\Gamma , \psi , \varphi \vdash \psi$}
    \RightLabel{\footnotesize \textsc{ex}}
    \UnaryInfC{$\Gamma , \varphi, \psi  \vdash \psi$}
    \RightLabel{\footnotesize \textsc{just}}
    \UnaryInfC{$\Gamma , \varphi, \psi  \vdash \psi \textit{ just}$}

    \RightLabel{\footnotesize $\EXI$E}
    \BinaryInfC{$\Gamma , \varphi \vdash \psi \textit{ just}$}

     \RightLabel{\footnotesize \textsc{w}}
    \UnaryInfC{$\Gamma , \varphi , \EXI \varphi \vdash \psi \textit{ just}$}

     \RightLabel{\footnotesize \textsc{ex}}
    \UnaryInfC{$\Gamma , \EXI \varphi , \varphi \vdash \psi \textit{ just}$}

    \RightLabel{\footnotesize $\EXI$E}
    \BinaryInfC{$\Gamma , \EXI \varphi \vdash \psi \textit{ just}$}

    \RightLabel{\footnotesize $\EXI$I}
    \UnaryInfC{$\Gamma , \EXI \varphi \vdash \EXI \psi $}

    \RightLabel{\footnotesize $\to$I}
    \UnaryInfC{$\Gamma \vdash \EXI \varphi \to \EXI \psi $}
\end{prooftree}

\begin{prooftree}
    \AxiomC{$\Gamma \vdash \EXI \varphi \to \EXI \psi \textit{ true}$}
    \RightLabel{\footnotesize \textsc{w}}
    \UnaryInfC{$\Gamma , \varphi \vdash \EXI \varphi \to \EXI \psi \textit{ true}$}

    \AxiomC{}
    \RightLabel{\footnotesize \textsc{hyp}}
    \UnaryInfC{$\Gamma , \varphi \vdash \varphi \textit{ true}$}
    \RightLabel{\footnotesize \textsc{just}}
    \UnaryInfC{$\Gamma , \varphi \vdash \varphi \textit{ just true}$}
    \RightLabel{\footnotesize $\EXI$I}
    \UnaryInfC{$\Gamma , \varphi \vdash \EXI \varphi \textit{ true}$}

    \RightLabel{\footnotesize $\to$E}
    
    \BinaryInfC{$\Gamma , \varphi \vdash \EXI \psi \textit{ true}$}

    \RightLabel{\footnotesize $\to$I}
    \UnaryInfC{$\Gamma \vdash \varphi \to \EXI \psi \textit{ true}$}
    
\end{prooftree}

\begin{proposition} \label{prop:3}
$\Gamma \vdash \varphi \to \EXI\psi \textit{ true}$ \,iff\, $\Gamma \vdash  \EXI\varphi \to \psi \textit{ just true}$.   
\end{proposition}

\noindent \textit{Proof}.

\begin{prooftree}
\def\defaultHypSeparation{\hskip 0.05in}
\small 
    \AxiomC{}
    \RightLabel{\footnotesize \textsc{hyp}}
    \UnaryInfC{$\Gamma , \EXI\varphi \textit{ true} \vdash \EXI\varphi \textit{ true}$}

    \AxiomC{$\Gamma \vdash \varphi \implyT \psi \textit{ true}$}
    \RightLabel{\footnotesize \textsc{w}}
    \UnaryInfC{$\Gamma , \varphi \textit{ true} \vdash \varphi \implyT \psi \textit{ true}$}
    
    \AxiomC{}
    \RightLabel{\footnotesize \textsc{hyp}}
    \UnaryInfC{$\Gamma , \varphi \textit{ true} \vdash \varphi \textit{ true}$}

    \RightLabel{\footnotesize $\implyT$E}
    \BinaryInfC{$\Gamma , \varphi \textit{ true} \vdash \psi \textit{ just true}$}

    \RightLabel{\footnotesize \textsc{w}}
    \UnaryInfC{$\Gamma , \varphi \textit{ true} , \EXI\varphi \textit{ true} \vdash \psi \textit{ just true}$}

    \RightLabel{\footnotesize \textsc{ex}}
    \UnaryInfC{$\Gamma , \EXI\varphi \textit{ true}, \varphi \textit{ true}  \vdash \psi \textit{ just true} $}
    \RightLabel{\footnotesize $\EXI$E}
    \BinaryInfC{$\Gamma , \EXI\varphi \textit{ true} \vdash \psi \textit{ just true}$}
    \RightLabel{\footnotesize $\to$Ij}
    \UnaryInfC{$\Gamma \vdash \EXI\varphi \to \psi \textit{ just true}$}
\end{prooftree}

\begin{prooftree}
    \AxiomC{$\Gamma \vdash \EXI\varphi \to \psi \textit{ just true}$}
    \RightLabel{\footnotesize \textsc{w}}
    \UnaryInfC{$\Gamma , \varphi \textit{ true} \vdash \EXI\varphi \to \psi \textit{ just true}$}
    \AxiomC{}
    \RightLabel{\footnotesize \textsc{hyp}}
    \UnaryInfC{$\Gamma , \varphi \textit{ true} \vdash \varphi \textit{ true}$}
    \RightLabel{\footnotesize \textsc{just}}
    \UnaryInfC{$\Gamma , \varphi \textit{ true} \vdash \varphi \textit{ just true}$}
    \RightLabel{\footnotesize $\EXI$I}
    \UnaryInfC{$\Gamma , \varphi \textit{ true} \vdash \EXI\varphi \textit{ true}$}
    \RightLabel{\footnotesize $\to$Ej}
    \BinaryInfC{$\Gamma , \varphi \textit{ true} \vdash \psi \textit{ just true}$}
    \RightLabel{\footnotesize $\implyT$I}
    \UnaryInfC{$\Gamma \vdash \varphi \implyT \psi \textit{ true}$}
\end{prooftree}

\begin{proposition}
$\Gamma \vdash \varphi \to \psi \textit{ just true}$ \,iff\, $\Gamma \vdash  \EXI\varphi \to \psi \textit{ just true}$.   
\end{proposition}

\noindent \textit{Proof}. Follows from Propositions \ref{prop:1} and \ref{prop:3}.

\begin{proposition} ($\EXI$ distributes over $\to$). 
If $\Gamma \vdash \EXI (\varphi \to \psi) \textit{ true}$ \,then\, $\Gamma \vdash  \EXI\varphi \to \EXI\psi \textit{ true}$.   
\end{proposition}

\noindent \textit{Proof}.

\begin{center}
\footnotesize
\def\defaultHypSeparation{\hskip 0.05in}
    \AxiomC{}
    \RightLabel{\footnotesize \textsc{hyp}}
    \UnaryInfC{$\Gamma , \EXI \varphi \vdash \EXI \varphi $}
    \AxiomC{$\Gamma \vdash \EXI (\varphi \to \psi) $}
    \RightLabel{\footnotesize \textsc{w}}
    \UnaryInfC{$\Gamma , \varphi \vdash \EXI (\varphi \to \psi) $}
    \AxiomC{}
    \RightLabel{\footnotesize \textsc{hyp}}
    \UnaryInfC{$\Gamma , \varphi \to \psi \vdash \varphi \to \psi  $}
    \RightLabel{\footnotesize \textsc{w}}
    \UnaryInfC{$\Gamma , \varphi \to \psi  , \varphi \vdash \varphi \to \psi  $}
    \RightLabel{\footnotesize \textsc{ex}}
    \UnaryInfC{$\Gamma , \varphi , \varphi \to \psi  \vdash \varphi \to \psi  $}
    \AxiomC{}
    \RightLabel{\footnotesize \textsc{hyp}}
    \UnaryInfC{$\Gamma , \varphi \vdash \varphi $}
    \RightLabel{\footnotesize \textsc{w}}
    \UnaryInfC{$\Gamma , \varphi , \varphi \to \psi \vdash \varphi $}
    \RightLabel{\footnotesize $\to$E}
    \BinaryInfC{$\Gamma , \varphi ,  \varphi \to \psi \vdash \psi $}
    \RightLabel{\footnotesize \textsc{just}}
    \UnaryInfC{$\Gamma , \varphi ,  \varphi \to \psi \vdash \psi \textit{ just}$}
    \RightLabel{\footnotesize $\EXI$E}
    \BinaryInfC{$\Gamma  , \varphi \vdash \psi \textit{ just}$}
    \RightLabel{\footnotesize \textsc{w}}
    \UnaryInfC{$\Gamma , \varphi , \EXI \varphi  \vdash \psi \textit{ just} $}
    \RightLabel{\footnotesize \textsc{ex}}
    \UnaryInfC{$\Gamma , \EXI \varphi, \varphi   \vdash \psi \textit{ just}$}
    \RightLabel{\footnotesize $\EXI$E}
    \BinaryInfC{$\Gamma , \EXI\varphi \vdash \psi \textit{ just}$}
    \RightLabel{\footnotesize $\EXI$I}
    \UnaryInfC{$\Gamma , \EXI\varphi \vdash \EXI \psi $}
    \RightLabel{\footnotesize $\to$I}
    \UnaryInfC{$\Gamma \vdash \EXI \varphi \to \EXI \psi $}
    \DisplayProof
\end{center}

\begin{proposition} \label{prop:6} 
$\Gamma \vdash \EXI \EXI \varphi \to \EXI \varphi  \textit{ true}$.    
\end{proposition}

\noindent \textit{Proof}.

\begin{prooftree}
\small
    \AxiomC{}
    \RightLabel{\footnotesize \textsc{hyp}}
    \UnaryInfC{$\Gamma , \EXI \EXI \varphi \vdash \EXI \EXI \varphi $}

    \AxiomC{}
    \RightLabel{\footnotesize \textsc{hyp}}
    \UnaryInfC{$\Gamma , \EXI \varphi \vdash \EXI \varphi $}
    
    \AxiomC{}
    \RightLabel{\footnotesize \textsc{hyp}}
    \UnaryInfC{$\Gamma , \varphi  \vdash \varphi $}
    \RightLabel{\footnotesize \textsc{just}}
    \UnaryInfC{$\Gamma , \varphi \vdash \varphi \textit{ just} $}
    \RightLabel{\footnotesize \textsc{w}}
    \UnaryInfC{$\Gamma , \varphi , \EXI\varphi \vdash \varphi \textit{ just} $}
    \RightLabel{\footnotesize \textsc{ex}}
    \UnaryInfC{$\Gamma , \EXI\varphi , \varphi  \vdash \varphi \textit{ just} $}
    
    \RightLabel{\footnotesize $\EXI$E}
    \BinaryInfC{$\Gamma , \EXI\varphi  \vdash \varphi \textit{ just}$}

    \RightLabel{\footnotesize \textsc{w}}
    \UnaryInfC{$\Gamma , \EXI\varphi , \EXI \EXI \varphi \vdash \varphi \textit{ just}$}

    \RightLabel{\footnotesize \textsc{ex}}
    \UnaryInfC{$\Gamma , \EXI \EXI \varphi , \EXI\varphi \vdash \varphi \textit{ just}$}

    \RightLabel{\footnotesize $\EXI$E}
    \BinaryInfC{$\Gamma , \EXI \EXI \varphi \vdash \varphi \textit{ just}$}

    \RightLabel{\footnotesize $\EXI$I}
    \UnaryInfC{$\Gamma , \EXI \EXI \varphi \vdash \EXI\varphi $}
    \RightLabel{\footnotesize $\to$I}
    \UnaryInfC{$\Gamma \vdash \EXI \EXI \varphi \to \EXI\varphi $}
\end{prooftree}

\section{A computational variant}
\label{sec:computational}

In this section, we introduce a computational interpretation of the system from Section \ref{sec:logical} with explicit proof expressions in the style of the Curry-Howard isomorphism. Specifically, instead of purely logical judgments $\varphi \textit{ true}$ and $\varphi \textit{ just true}$ we will work with judgments $e : \varphi$ and $e \existsJ \varphi$ where $e$ is a proof expression representing the proof of $\varphi$. The informal meaning of $e : \varphi$ and $e \existsJ \varphi$ remains the same, i.e., ``$\varphi$ is true'' and ``there exists a proof of $\varphi$'', respectively. Additionally, we introduce new forms of judgment for expressing the notions of reducibility $\Longrightarrow$ and computability $\mathbf{R}$ of expressions $e$ of type $\varphi$.

The computational variant satisfies all the discussed properties of the logical variant, including the substitution lemma. Furthermore, we show that the computational variant has the subject reduction property and other useful properties such as forward and backward preservation, which then allows us to finally show that it also has the strong normalization property.

\subsection{Formal system}
\label{sec:formal_comp}

Language:

\begin{align*}
\text{Propositions } \varphi, \psi & \;::=\; p \mid \varphi \to \psi \mid\EXI\varphi \\
\text{Terms } t, t' & \;::=\; x \mid \lambda x.t \mid \textsf{ap}(t, t') \mid \langle e \rangle  \\
\text{Expressions } e, e' & \;::=\; t \mid \textsf{let } \langle x \rangle \textsf{ be } t \textsf{ in } e \mid  \lambda_{\textrm{j}} x.e \mid \textsf{ap}_{\textrm{j}}(e, t') \mid \langle e \rangle_{\textrm{j}} \mid  \textsf{let } \langle x \rangle_{\textrm{j}} \textsf{ be } e' \textsf{ in } e  \\ 
\text{Hypotheses } \Gamma & \;::=\; \cdot \mid \Gamma, x : \varphi \\
\end{align*}

\noindent \textit{Note}. All terms are expressions but not vice versa. The main difference between terms and expressions is that any two expressions $e$ and $e'$ of type $\varphi$ are treated as (propositionally) equal, which is not the case for terms. This is a consequence of the fact existential judgments capture the notion of proof-irrelevant truth.

\medskip

\noindent Categorical judgments:
$$t : \varphi \qquad e \existsJ \varphi $$

\noindent Hypothetical judgments:

$$\Gamma \vdash t : \varphi \qquad \Gamma \vdash e \existsJ \varphi $$

\medskip
\noindent \textit{Note}.
\noindent The traditional logical reading of the symbol ``$\existsJ$'' as ``therefore'' can still be applied here, i.e., it still makes sense to read $e \existsJ \varphi$ as ``$e$, therefore $\varphi$'' -- there exists a proof $e$ for $\varphi$, \emph{therefore} we can claim that $\varphi$ is true. Also, from a type-theoretic view, one reads $t : \varphi$ as ``$t$ has type $\varphi$'' or ``$t$ is typed to $\varphi$''. Analogously, $t \existsJ \varphi$ can be read as ``there exists a term of type $\varphi$'' or ``$t$ is irrelevantly typed to $\varphi$''. Or, from a proof-theoretic view, ``there exists a proof of proposition $\varphi$''.

\medskip

Similarly to the metavariable $\textit{star}$ standing either for $\textit{true}$ or $\textit{just true}$ from Section \ref{sec:logical}, we use the metavariable $\starJ$ standing either for ``$:$'' or ``$\existsJ$'', i.e., either proof-relevant or proof-irrelevant typing (again, all occurrences in rules, theorems, etc. must be consistent): $$\Gamma \vdash e \starJ \varphi$$ Also, note that since all terms $t$ are expressions $e$ but not vice versa, instantiating judgment $e \starJ \varphi$ with ``$:$'' leads to $t : \varphi$, not $e : \varphi$. In other words, in cases where $\starJ$ is replaced with ``$:$'', expressions $e$ are restricted to terms $t$.

\medskip

Additionally, to the previous section, we add two new forms of judgments:

$$\Gamma \vdash e \Longrightarrow_R e' \qquad \Gamma \vdash e \Longrightarrow_E e'$$

\noindent The first judgment, $\Gamma \vdash e \Longrightarrow e'$, tells us that in the context $\Gamma$, the expression $e$ reduces to the expression $e'$. The second judgment, $\Gamma \vdash e \Longrightarrow_E e'$, tells us that in the context $\Gamma$, the expression $e$ expands to the expression $e'$. In a computational setting, thanks to the Curry-Howard correspondence, these new judgments on proof expressions correspond to the local reduction and expansion metarules $\Rightarrow_R$ and $\Rightarrow_E$ from previous Section \ref{sec:logical}. To simplify the notation a little, we will omit the subscript ``$R$'' with reductions and write simply $\Longrightarrow$, while keeping the subscript ``$E$'' for expansions.

\medskip

The above judgments are governed by the following rules:

\begin{center}
\AxiomC{}
\RightLabel{\footnotesize \textsc{hyp}}
\UnaryInfC{$\Gamma, x : \varphi \vdash x : \varphi   $}
\DisplayProof
\quad
\AxiomC{$\Gamma \vdash t : \varphi $}
\RightLabel{\footnotesize \textsc{just}}
\UnaryInfC{$\Gamma \vdash t \existsJ \varphi$}
\DisplayProof
\end{center}

\medskip

\noindent To make the presentation of logical rules more surveyable, we skip the schematic  $\textit{star}$ variants $\to$I$\starJ$, $\to$E$\starJ$, $\EXI$I$\starJ$, $\EXI$E$\starJ$ and present directly their typed variants instantiated by ``$:$'' and ``$\existsJ$'', respectively (we will write this as $\starJ$ = $:$ and $\starJ$ = $\existsJ$), i.e., $\to$I, $\to$E, $\EXI$I, $\EXI$E and $\to$Ij, $\to$Ej, $\EXI$Ij, $\EXI$Ej. We also add corresponding computations (C) and expansion (CE) rules. 

First, let us present the $\textit{true}$ variants (i.e., $\starJ$ = $:$):

\medskip

\begin{center}
\AxiomC{$\Gamma, x : \varphi \vdash t : \psi $}
\RightLabel{\footnotesize $\to$I}
\UnaryInfC{$\Gamma \vdash \lambda x . t : \varphi \to \psi $}
\DisplayProof
\quad
\AxiomC{$\Gamma \vdash t : \varphi \to \psi $}
\AxiomC{$\Gamma \vdash t' : \varphi  $}
\RightLabel{\footnotesize $\to$E}
\BinaryInfC{$\Gamma \vdash  \textsf{ap}(t, t') : \psi $}
\DisplayProof
\end{center}

\begin{center}
\AxiomC{$\Gamma , x : \varphi \vdash t : \psi$}
\AxiomC{$\Gamma \vdash t' : \varphi$}
\RightLabel{\footnotesize $\to$C}
\BinaryInfC{$\Gamma \vdash \textsf{ap}(\lambda x . t , t') \Longrightarrow [t'/x] t $}
\DisplayProof
\quad
\AxiomC{$\Gamma \vdash t : \varphi \to \psi $}
\RightLabel{\footnotesize $\to$CE}
\UnaryInfC{$\Gamma \vdash t \Longrightarrow_E \lambda x . \textsf{ap}(t, x)$}
\DisplayProof
\end{center}

\begin{prooftree}
    \AxiomC{$\Gamma \vdash e \Longrightarrow e' $}
    \RightLabel{\footnotesize \textsf{ap}C}
    \UnaryInfC{$\Gamma \vdash \textsf{ap}(e, e'') \Longrightarrow \textsf{ap}(e', e'') $}
\end{prooftree}

\begin{center}
\AxiomC{$\Gamma \vdash e \existsJ \varphi $}
\RightLabel{\footnotesize $\EXI$I}
\UnaryInfC{$\Gamma \vdash \langle e \rangle : \EXI\varphi  $}
\DisplayProof
\quad
\AxiomC{$\Gamma \vdash t : \EXI\varphi$}
\AxiomC{$\Gamma , x : \varphi \vdash e \existsJ \gamma $}
\RightLabel{\footnotesize $\EXI$E}
\BinaryInfC{$\Gamma \vdash \textsf{let } \langle x \rangle \textsf{ be } t \textsf{ in } e  \existsJ \gamma  $}
\DisplayProof
\end{center}

\begin{center}
\small
\AxiomC{$\Gamma , x : \varphi \vdash e \existsJ \gamma$}
\AxiomC{$\Gamma \vdash t : \EXI\varphi$}
\RightLabel{\footnotesize $\EXI$C}
\BinaryInfC{$\Gamma \vdash \textsf{let } \langle x \rangle \textsf{ be } \langle e' \rangle \textsf{ in } e \Longrightarrow \Sopen e'/x\Sclose  e $}
\DisplayProof
\medskip

\medskip

\quad
\AxiomC{$\Gamma \vdash t : \EXI\varphi$}
\RightLabel{\footnotesize $\EXI$CE}
\UnaryInfC{$\Gamma \vdash t \Longrightarrow_E \langle \textsf{let } \langle x \rangle  \textsf{ be } t \textsf{ in }  x  \rangle $}
\DisplayProof
\end{center}

\medskip 

\noindent Now, let us present the $\textit{just true}$ variant (i.e., $\starJ$ = $\existsJ$):

\medskip

\begin{center}
\AxiomC{$\Gamma, x : \varphi \vdash e \existsJ \psi $}
\RightLabel{\footnotesize $\to$Ij}
\UnaryInfC{$\Gamma \vdash \lambda_{\textrm{j}} x . e \existsJ \varphi \to \psi $}
\DisplayProof
\quad
\AxiomC{$\Gamma \vdash e \existsJ \varphi \to \psi $}
\AxiomC{$\Gamma \vdash t : \varphi  $}
\RightLabel{\footnotesize $\to$Ej}
\BinaryInfC{$\Gamma \vdash  \textsf{ap}_{\textrm{j}}(e, t) \existsJ \psi $}
\DisplayProof
\end{center}

\begin{center}
\AxiomC{$\Gamma , x : \varphi \vdash e \existsJ \psi$}
\AxiomC{$\Gamma \vdash t : \varphi$}
\RightLabel{\footnotesize $\to$Cj}
\BinaryInfC{$\Gamma \vdash \textsf{ap}_{\textrm{j}}(\lambda_{\textrm{j}} x . e , t) \Longrightarrow [t/x] e $}
\DisplayProof
\quad
\AxiomC{$\Gamma \vdash e \existsJ \varphi \to \psi $}
\RightLabel{\footnotesize $\to$CEj}
\UnaryInfC{$\Gamma \vdash e \Longrightarrow_E \lambda_{\textrm{j}} x . \textsf{ap}_{\textrm{j}}(e, x)$}
\DisplayProof
\end{center}

\begin{prooftree}
    \AxiomC{$\Gamma \vdash e \Longrightarrow e' $}
    \RightLabel{\footnotesize \textsf{ap}Cj}
    \UnaryInfC{$\Gamma \vdash \textsf{ap}_{\textrm{j}}(e, e'') \Longrightarrow \textsf{ap}_{\textrm{j}}(e', e'') $}
\end{prooftree}

\begin{center}
\AxiomC{$\Gamma \vdash e \existsJ \varphi $}
\RightLabel{\footnotesize $\EXI$Ij}
\UnaryInfC{$\Gamma \vdash \langle e \rangle_{\textrm{j}} \existsJ \EXI\varphi  $}
\DisplayProof
\quad
\AxiomC{$\Gamma \vdash e' \existsJ \EXI\varphi$}
\AxiomC{$\Gamma , x : \varphi \vdash e \existsJ \gamma $}
\RightLabel{\footnotesize $\EXI$Ej}
\BinaryInfC{$\Gamma \vdash \textsf{let } \langle x \rangle_{\textrm{j}} \textsf{ be } e' \textsf{ in } e  \existsJ \gamma  $}
\DisplayProof
\end{center}

\begin{center}
\small
\AxiomC{$\Gamma , x : \varphi \vdash e \existsJ \gamma$}
\AxiomC{$\Gamma \vdash e'' \existsJ \EXI\varphi$}
\RightLabel{\footnotesize $\EXI$Cj}
\BinaryInfC{$\Gamma \vdash \textsf{let } \langle x \rangle_{\textrm{j}} \textsf{ be } \langle e' \rangle_{\textrm{j}} \textsf{ in } e \Longrightarrow \Sopen e'/x\Sclose  e $}
\DisplayProof
\medskip

\medskip

\quad
\AxiomC{$\Gamma \vdash e \existsJ \EXI\varphi$}
\RightLabel{\footnotesize $\EXI$CEj}
\UnaryInfC{$\Gamma \vdash e \Longrightarrow_E \langle \textsf{let } \langle x \rangle_{\textrm{j}}  \textsf{ be } e \textsf{ in }  x  \rangle $}
\DisplayProof
\end{center}

Note that the newly introduced \textit{just true} variants of canonical and noncanonical proof expressions, namely $\lambda_{\textrm{j}} x.e, \textsf{ap}_{\textrm{j}}(e, e')$, $\langle e \rangle_{\textrm{j}}$, $\textsf{let } \langle x \rangle_{\textrm{j}} \textsf{ be } e' \textsf{ in } e$, behave essentially identically to their \textit{true} variants, only the associated typing relation changes from the proof-relevant $:$ to the proof-irrelevant $\existsJ$\,. For example, in the proof-irrelevant application $\textsf{ap}_{\textrm{j}}(e, t')$ the first argument is an expression, while in the proof-relevant application $\textsf{ap}(t, t')$ the first argument is a term.

\subsection{Basic properties}

\subsubsection{Substitution}

\medskip

\begin{center}
\AxiomC{$\Gamma \vdash e \starJ \varphi $}
\AxiomC{$\Gamma, x : \varphi \vdash e' \starJ  \gamma $}
\RightLabel{\footnotesize \textsc{sub1}}
\BinaryInfC{$\Gamma \vdash \Sopen e/x \Sclose e' \starJ  \gamma $}
\DisplayProof \quad
\AxiomC{$\Gamma \vdash t : \varphi$}
\AxiomC{$\Gamma, x : \varphi \vdash e  \existsJ  \gamma $}
\RightLabel{\footnotesize \textsc{sub2}}
\BinaryInfC{$\Gamma \vdash [ t /x ] e  \existsJ  \gamma $}
\DisplayProof
\end{center}

\noindent where $[t/x]e$ from \textsc{sub2} represents the result of substituting $t$ for $x$ in $e$ while renaming bound variables where necessary to avoid variable capture. The definition of the new kind of substitution $\Sopen e/x \Sclose e'$ from \textsc{sub1} is a bit more tricky. When $\starJ$ = $:$, the substitution $ \Sopen e/x \Sclose e'$ becomes just $[t/x]t'$. However, when $\starJ$ = $\existsJ$, the substitution $[e/x]e' $ behaves differently from the standard substitution because there occurs a mismatch between what is being substituted for what. Specifically, when $\starJ$ = $\existsJ$, we are trying to substitute $e$ that is typed irrelevantly (i.e., $e \existsJ \varphi$) for a variable $x$ that is typed relevantly (i.e., $x : \varphi$). Consequently, to account for this mismatch, the substitution has to be defined inductively on the structure of $e$ not of $e'$ (we base this on \cite{pfenning2001}):

$$ \Sopen e /x \Sclose e' =_{df}
  \begin{cases}
    [t/x]e' & \text{if } e = t \\
    \textsf{let } \langle y \rangle \textsf{ be } t \textsf{ in } \Sopen e'' /x\Sclose e' & \text{if } e = \textsf{let } \langle y \rangle  \textsf{ be } t \textsf{ in } e''\\
    \textsf{let } \langle y \rangle_{\textrm{j}} \textsf{ be } e''' \textsf{ in } \Sopen e'' /x\Sclose e' & \text{if } e = \textsf{let } \langle y \rangle_{\textrm{j}}  \textsf{ be } e''' \textsf{ in } e''\\
     \langle \Sopen e'' / x \Sclose e' \rangle  & \text{if } e = \langle e'' \rangle_{\textrm{j}} \\
    \lambda_{\textrm{j}} x . \Sopen e'' / x \Sclose e'  & \text{if } e = \lambda_{\textrm{j}} x . e'' \\
    \textsf{ap}_{\textrm{j}}( \Sopen e'' / x \Sclose e',t)  & \text{if } e = \textsf{ap}_{\textrm{j}}(e'',t) \\

  \end{cases}
$$

Thus, when expression $e$ is a term $t$, $\Sopen e /x \Sclose e'$ reduces to the standard substitution $[t/x]e'$. Otherwise, i.e., if $e$ is an expression of any of the specified forms, there is a recursive call of the substitution as specified by the clauses.

Hence, in practice, we will have the following three substitution rules:

\begin{prooftree}
\AxiomC{$\Gamma \vdash t : \varphi $}
\AxiomC{$\Gamma, x : \varphi \vdash t': \gamma $}
\RightLabel{\footnotesize \textsc{sub1} ($\starJ$ = $:$)}
\BinaryInfC{$\Gamma \vdash [t/x]t' : \gamma $}
\end{prooftree}

\begin{prooftree}
\AxiomC{$\Gamma \vdash e \existsJ \varphi$}
\AxiomC{$\Gamma, x : \varphi \vdash e' \existsJ \gamma $}
\RightLabel{\footnotesize \textsc{sub1} ($\starJ$ = $\existsJ$)}
\BinaryInfC{$\Gamma \vdash \Sopen e /x \Sclose e'  \existsJ  \gamma $}
\end{prooftree}

\begin{prooftree}
\AxiomC{$\Gamma \vdash t : \varphi$}
\AxiomC{$\Gamma, x : \varphi \vdash e  \existsJ  \gamma $}
\RightLabel{\footnotesize \textsc{sub2}}
\BinaryInfC{$\Gamma \vdash [ t /x ] e  \existsJ  \gamma $}
\end{prooftree}

\medskip

\noindent We can show that the substitution rules \textsc{sub1} and \textsc{sub2} are valid using the following lemma.

\begin{lemma} (Substitution lemma)

\begin{enumerate}

    \item[](\textsc{sub1}) If $\Gamma \vdash e \starJ \varphi $ and $\Gamma, x : \varphi \vdash e' \starJ  \gamma $ then $\Gamma \vdash \Sopen e/x \Sclose e' \starJ  \gamma $.

    \item[](\textsc{sub2}) If $\Gamma \vdash t : \varphi$ and $\Gamma, x : \varphi \vdash e  \existsJ  \gamma $ then $\Gamma \vdash [ t /x ] e  \existsJ  \gamma $.

\end{enumerate}   
\end{lemma}

\begin{proof}
(\textsc{sub1}, $\star$ = $:$) By structural induction on the derivation of $\Gamma , x : \varphi \vdash t' : \gamma $.   
\end{proof}

\begin{proof}
(\textsc{sub1}, $\star$ = $\existsJ$) By structural induction on the derivation of $\Gamma \vdash e \existsJ \varphi$.
\end{proof}

\begin{proof}
(\textsc{sub2}) By structural induction on the derivation of $\Gamma , x : \varphi \vdash e \existsJ \gamma $.   
\end{proof}

\subsubsection{Subject reduction}

\begin{lemma} (Subject reduction)
If $\Gamma \vdash e \starJ \varphi$ and $\Gamma \vdash e \Longrightarrow e'$ then $\Gamma \vdash e' \starJ \varphi$.
\end{lemma}

\begin{proof}
By structural induction on the derivation $\Gamma \vdash e \Longrightarrow e'$ and by inversion on $\Gamma \vdash e \starJ \varphi$.   
\end{proof}

\noindent For example, let us consider the case of computation rule for $\EXI$:

\begin{prooftree}
\AxiomC{$\Gamma , x : \varphi \vdash e \existsJ \gamma$}
\AxiomC{$\Gamma \vdash t : \EXI\varphi$}
\RightLabel{\footnotesize $\EXI$C}
\BinaryInfC{$\Gamma \vdash \textsf{let } \langle x \rangle \textsf{ be } \langle e' \rangle \textsf{ in } e \Longrightarrow \Sopen e' /x\Sclose  e $}
\end{prooftree}

\noindent where $e = \textsf{let } \langle x \rangle \textsf{ be } \langle e' \rangle \textsf{ in } e \existsJ \varphi$.

\begin{enumerate}
    \item $\Gamma \vdash \textsf{let } \langle x \rangle \textsf{ be } \langle e' \rangle \textsf{ in } e \existsJ \varphi$ (by assumption)
    \item $\Gamma \vdash \langle e' \rangle : \EXI\psi$ and $\Gamma , x : \psi \vdash e \existsJ \varphi$ (by inversion on 1.)
    \item $\Gamma \vdash e' \existsJ \psi$ (by inversion on 2.)
    \item $\Gamma \vdash \Sopen e' /x\Sclose  e \existsJ \varphi $ (by \textsc{sub1} on 3. and 2.)
\end{enumerate}

\subsubsection{Normalization}

We prove strong normalization, i.e., the property that every reduction of an expression ends with a value, by Tait's reducibility method (\cite{tait1967}), also known as the method of logical relations (\cite{plotkin1973}).

First, let us introduce a new form of judgment: $$\Gamma \vdash e \star \varphi \in \mathbf{R}$$ This judgment means that in the context $\Gamma$ expression $e$ (which is typed either relevantly $e : \varphi$ or irrelevantly $e \existsJ \varphi$) belongs to the set $\mathbf{R}$ of reducible expressions of type $\varphi$.\footnote{Typically, these sets are regarded as unary logical relations, i.e., predicates on terms and written as $\mathbf{R}_\varphi (e)$. However, since our system additionally distinguishes between relevant and irrelevant typing, i.e., between $t : \varphi$ and $e \existsJ \varphi$, we use this nonstandard notation to keep the typing information explicit.}

The set \textbf{R}, or more precisely, the judgment of the form $\Gamma \vdash e \starJ \varphi \in \mathbf{R}$, is defined for $\EXI\varphi$ and $\varphi \to \psi$ as follows (recall that $\varphi \implyT \psi$ is defined as $\varphi \to \EXI\psi$):\footnote{Note that the property $\mathbf{R}$ of reducible expressions is inductively defined over the structure of $\EXI\varphi$, so it is well-founded.}

\begin{itemize}

    \item $\Gamma \vdash e \starJ \EXI\varphi \in \mathbf{R}$ \, $=_{df}$ \, $\Gamma \vdash e \starJ \EXI\varphi$ and $\Gamma \vdash e \Longrightarrow \langle a \rangle$ for some proof term $a$ such that $\Gamma \vdash a : \varphi \in \mathbf{R}$. In a rule form:
    
    \begin{prooftree}
        \AxiomC{$\Gamma \vdash e \starJ \EXI\varphi$}
        \AxiomC{$\Gamma \vdash e \Longrightarrow \langle a \rangle$}
        \AxiomC{$\Gamma \vdash a : \varphi \in \mathbf{R}$}
        \TrinaryInfC{$\Gamma \vdash e \starJ \EXI\varphi \in \mathbf{R}$}
    \end{prooftree}

    \item $\Gamma \vdash e \starJ \varphi \to \psi \in \mathbf{R}$ \, $=_{df}$ \, $\Gamma \vdash e \starJ \varphi \to \psi$ and $\Gamma \vdash e \Longrightarrow \lambda x . t$ and there exists a term $t$ such that for all $a$, $\Gamma \vdash a : \varphi \in \mathbf{R}$ implies $\Gamma \vdash \textsf{ap}(t,a) \starJ \psi \in \mathbf{R}$. In a rule form:

    \begin{prooftree}
        \AxiomC{$\Gamma \vdash e \starJ \varphi \to \psi$}
        \AxiomC{$\Gamma \vdash e \Longrightarrow \lambda x . t$}
        \AxiomC{$\Gamma , x : \varphi \vdash t \starJ \psi $}
        \TrinaryInfC{$\Gamma \vdash e \starJ \varphi \to \psi \in \mathbf{R}$}
    \end{prooftree}

    \noindent with the additional condition on term $t$ as above.
    
\end{itemize}

\noindent The proof of the normalization proceeds in two steps. First, we show that if expression $e$ (typed either relevantly or irrelevantly to $\varphi$) has the property $\textbf{R}$ in context $\Gamma$, then the reduction of $e$ terminates, i.e., ends with a value (Lemma 7). Then, we show that if we substitute for the free variables of $e$ terms that have the property $\textbf{R}$, then the result will have the property $\textbf{R}$ as well (Lemma 8).

\begin{lemma}(Existence of value)
If $\Gamma \vdash e \starJ \varphi \in \mathbf{R}$, then $\Gamma \vdash v \existsJ \varphi$, i.e., there exists a canonical proof expression, i.e., value $v$ of the type $\varphi$, such that $\Gamma \vdash e \Longrightarrow v$.   
\end{lemma}

\begin{proof}
Follows from the definition of judgments of the form $\Gamma \vdash e \starJ \varphi \in \mathbf{R}$.   
\end{proof}

\begin{lemma} (Fundamental lemma)
If $\Gamma \vdash e(x_1 , \ldots , x_n) \starJ \varphi$, where $\Gamma = x_1 : \gamma_1 , \ldots, x_n : \gamma_n$, and $t_1, \ldots, t_n$ are terms such that $\Gamma \vdash t_i : \gamma_i \in \mathbf{R}$, $0 < i \leq n$, then $\Gamma \vdash e(t_1 , \ldots , t_n) \starJ \varphi \in \mathbf{R}$.   
\end{lemma}

\noindent where $\starJ$ stands for either ``$:$'' or ``$\existsJ$'' symbol (all occurencens must be consistent; also, recall that all terms $t$ are expressions $e$ but not vice versa).

\begin{proof}
By structural induction on the derivation $\Gamma \vdash e(x_1 , \ldots , x_n) \star \varphi$ (with the help of preservation and backward preservation lemmas, see below).    
\end{proof}

\noindent For example:

\medskip

\noindent Case $\EXI$I (where $\starJ$ = $:$):

\begin{center}
\AxiomC{$\Gamma \vdash e(x_1 , \ldots , x_n) \existsJ \varphi$}
\RightLabel{\footnotesize $\EXI$I}
\UnaryInfC{$\Gamma \vdash \langle e(x_1 , \ldots , x_n) \rangle : \EXI \varphi$}
\DisplayProof
\end{center}

\begin{enumerate}

    \item $\Gamma \vdash t_1 : \gamma_1  \in \mathbf{R} , \ldots, t_n : \gamma_n  \in \mathbf{R}$ (by assumption)

    \item $\Gamma \vdash  e(t_1, \ldots, t_n) \existsJ \varphi \in \mathbf{R}  $ (by IH)

    \item $\Gamma \vdash \langle e(t_1, \ldots, t_n) \rangle : \EXI\varphi \in \mathbf{R} $ (by def. of $t : \EXI\varphi \in \mathbf{R}$)
    
\end{enumerate}

\noindent Case $\EXI$E (where $\star$ = $:$):

\begin{center}
\AxiomC{$\Gamma \vdash t(x_1 , \ldots , x_n) : \EXI \varphi$}
\AxiomC{$\Gamma , x : \varphi \vdash e(x_1 , \ldots , x_n, x) \existsJ \gamma$}
\RightLabel{\footnotesize $\EXI$E}
\BinaryInfC{$\Gamma \vdash \textsf{let } \langle x \rangle \textsf{ be } t(x_1 , \ldots , x_n) \textsf{ in } e(x_1 , \ldots , x_n) \existsJ \gamma$}
\DisplayProof
\end{center}

\begin{enumerate}

    \item $\Gamma \vdash t(t_1 , \ldots , t_n) : \EXI\varphi \in \mathbf{R}$ (by IH)

    \item $\Gamma \vdash e(t_1 , \ldots , t_n, x) \existsJ \gamma \in \mathbf{R}$ for all $x$ such that $\Gamma \vdash x : \varphi \in \mathbf{R}$ (by IH)

    \item $\Gamma \vdash t(t_1 , \ldots , t_n) \Longrightarrow \langle a \rangle : \EXI\varphi $ for some $a$ such that $\Gamma \vdash a : \varphi \in \mathbf{R}$ (by def. of $t : \EXI\varphi \in \mathbf{R}$ from 1.)

    \item $\Gamma \vdash \textsf{let } \langle x \rangle \textsf{ be } \langle a \rangle \textsf{ in } e(t_1 , \ldots , t_n, x) \Longrightarrow \Sopen a/x\Sclose  e(t_1 , \ldots , t_n, x) $ (by def. of $\textsf{let}$)

    \item $\Gamma \vdash \textsf{let } \langle x \rangle \textsf{ be } \langle a \rangle \textsf{ in } e(t_1 , \ldots , t_n, x) \Longrightarrow e(t_1 , \ldots , t_n, a) $ (by def. of sub1 on 4.)

    \item $\Gamma \vdash e(t_1 , \ldots , t_n, a) \existsJ \gamma \in \mathbf{R}$  (by IH from 2. and 3.)

    \item $\Gamma \vdash \textsf{let } \langle x \rangle \textsf{ be } t(t_1 , \ldots , t_n) \textsf{ in } e(t_1 , \ldots , t_n) \existsJ \gamma \in \mathbf{R} $ (by 3., 4, 5., and backward preservation lemma)
    
\end{enumerate}

\medskip

\noindent Case $\to$Ij (i.e., $\to$I$\starJ$ where $\starJ$ = $\existsJ$):

\begin{prooftree}
\AxiomC{$\Gamma, x : \varphi \vdash e(x_1 , \ldots , x_n , x) \existsJ \psi $}
\RightLabel{\footnotesize $\to$Ij}
\UnaryInfC{$\Gamma \vdash \lambda_{\textrm{j}} x . e(x_1 , \ldots , x_n , x) \existsJ \varphi \to \psi $}
\end{prooftree}

\begin{enumerate}

    \item $\Gamma \vdash t_1 : \gamma_1  \in \mathbf{R} , \ldots, t_n : \gamma_n  \in \mathbf{R}$ (by assumption)    
    
    \item $\Gamma \vdash e(t_1 , \ldots , t_n , x) \existsJ \psi  \in \mathbf{R}$  for all $x$ such that $\Gamma \vdash x : \varphi \in \mathbf{R}$. (by IH)
    
    \item  $\Gamma \vdash a : \varphi \in \mathbf{R}$ for some arbitrary $a$ (by assumption)
    
    \item $\Gamma \vdash a \Longrightarrow v$ for some $v$ (by def. of $\mathbf{R}$ from 3.) 
    
    \item $\Gamma \vdash  v : \varphi \in \mathbf{R} $ (by preservation lemma on 4.)
    
    \item $\Gamma \vdash e(t_1 , \ldots , t_n , v) \existsJ \psi  \in \mathbf{R}$  (by IH from 2. and 5.)
    
    \item $\Gamma \vdash \textsf{ap}_{\textrm{j}} (\lambda_{\textrm{j}} x . e(t_1 , \ldots , t_n , x) , v) \Longrightarrow e(t_1 , \ldots , t_n , v) $ (by def. of $\to$Cj)
    
    \item $\Gamma \vdash \textsf{ap}_{\textrm{j}} (\lambda_{\textrm{j}} x . e(t_1 , \ldots , t_n , x) , v) \existsJ \psi \in \mathbf{R} $ (by backward preservation lemma from 7.)
    
    \item $\Gamma \vdash \lambda_{\textrm{j}} x . e(t_1 , \ldots , t_n , x)  \existsJ \varphi \to \psi \in \mathbf{R} $ (by def. of $\varphi \to \psi \in \mathbf{R}$  from 8. since $a$ was chosen arbitrarily)

\end{enumerate}

\noindent Case $\to$Ej (i.e., $\to$E$\starJ$ where $\starJ$ = $\existsJ$): 

\begin{prooftree}
\AxiomC{$\Gamma \vdash e(x_1 , \ldots , x_n) \existsJ \varphi \to \psi $}
\AxiomC{$\Gamma \vdash t(x_1 , \ldots , x_n) : \varphi  $}
\RightLabel{\footnotesize $\to$Ej}
\BinaryInfC{$\Gamma \vdash  \textsf{ap}_{\textrm{j}}(e(x_1 , \ldots , x_n), t(x_1 , \ldots , x_n)) \existsJ \psi $}
\end{prooftree}

\begin{enumerate}

    \item $\Gamma \vdash e(t_1 , \ldots , t_n) \existsJ \varphi \to \psi \in \mathbf{R}$ (by IH)

    \item $\Gamma \vdash t(t_1 , \ldots , t_n) : \varphi \in \mathbf{R}$ (by IH)

    \item $\Gamma \vdash  \textsf{ap}_{\textrm{j}}(e(t_1 , \ldots , t_n), t(t_1 , \ldots , t_n)) \existsJ \psi \in \mathbf{R}$ (by def. of $e \starJ \varphi \to \psi \in \mathbf{R}$ on 1., specifically we use the additional condition instantiated with 2.)

\end{enumerate}

\begin{theorem} (Strong normalization)
If $\Gamma \vdash e \starJ \varphi$, then $\Gamma \vdash e \Longrightarrow v $. 
\end{theorem}

\begin{proof}
Follows from Lemmas 7 and 8. By choosing the context to be empty, i.e., by taking the expression $e$ to be closed (with no free variables), we obtain that $e$ is normalizable whenever $\vdash e \starJ \varphi$.    
\end{proof}

\begin{lemma} (Backward preservation)
If $\Gamma \vdash e' \starJ \varphi \in \mathbf{R}$ and $\Gamma \vdash e \Longrightarrow e'$ then $\Gamma \vdash e \starJ \varphi \in \mathbf{R}$.   
\end{lemma}

\begin{proof}
By induction on the judgment $\Gamma \vdash e \Longrightarrow e'$.    
\end{proof}

\noindent For example:

\medskip

\noindent Case $\textsf{let}$:

\begin{prooftree}
\AxiomC{$\Gamma , x : \varphi \vdash e_2 \existsJ \gamma$}
\AxiomC{$\Gamma \vdash t : \EXI\varphi$}
\RightLabel{\footnotesize $\EXI$C}
\BinaryInfC{$\Gamma \vdash \textsf{let } \langle x \rangle \textsf{ be } \langle e_1 \rangle \textsf{ in } e_2 \Longrightarrow \Sopen e_1 /x\Sclose  e_2 $}
\end{prooftree}

\medskip

\noindent Let $e = \textsf{let } \langle x \rangle \textsf{ be } \langle e_1 \rangle \textsf{ in } e_2$ and $e' = \Sopen e_1 / x \Sclose e_2$.

\medskip

\begin{enumerate}

    \item $\Gamma \vdash \textsf{let } \langle x \rangle \textsf{ be } \langle e_1 \rangle \textsf{ in } e_2 \Longrightarrow \Sopen e_1 / x \Sclose e_2$ (by assumption)
    
    \item $\Gamma \vdash \Sopen e_1 / x \Sclose e_2 \starJ \varphi \in \mathbf{R}$ (by assumption)
    
\end{enumerate}

\noindent Now, in general, to be able to claim that $\Gamma \vdash e \starJ \varphi \in \mathbf{R}$, i.e., that $e$ is reducible, we have to show that all possible reductions of $e$ result in terms that are also reducible. In other words, we need to show that to whatever $e''$ the $e$ reduces, it has the $\mathbf{R}$ property. So, let us assume: 

\begin{enumerate}
    \item[3.] $\textsf{let } \langle x \rangle \textsf{ be } \langle e_1 \rangle \textsf{ in } e_2 \Longrightarrow e''$ (by assumption)
\end{enumerate}

\noindent Since $\textsf{let}$ has only one reduction rule, i.e., $\textsf{let } \langle x \rangle \textsf{ be } \langle e_1 \rangle \textsf{ in } e_2 \Longrightarrow \Sopen e_1 / x \Sclose e_2$, it means that $e'' = e' = \Sopen e_1 / x \Sclose e_2$. And, by assumption 2, we know that $\Sopen e_1 / x \Sclose e_2 \starJ \varphi \in \mathbf{R}$. Thus, we can conclude that:

\begin{enumerate}
    \item[4.] $\Gamma \vdash \textsf{let } \langle x \rangle \textsf{ be } \langle e_1 \rangle \textsf{ in } e_2 \starJ \varphi \in \mathbf{R}$
\end{enumerate}

\begin{lemma} (Preservation)
If $\Gamma \vdash e \starJ \varphi \in \mathbf{R}$ and $\Gamma \vdash e \Longrightarrow e'$ then $\Gamma \vdash e' \starJ \varphi \in \mathbf{R}$.    
\end{lemma}

\begin{proof}
By induction on the structure of the type $\varphi$.   
\end{proof}

\noindent For example:

\medskip

\noindent Case $\varphi = \EXI\varphi $:

\begin{enumerate}
    \item $\Gamma \vdash e : \EXI\varphi \in \mathbf{R}$ (by assumption)

    \item $\Gamma \vdash e \Longrightarrow \langle a \rangle $ for some $\Gamma \vdash a : \varphi \in \mathbf{R}$ (by def. of $\EXI\varphi \in \mathbf{R}$ from 1.)

    \item $\Gamma \vdash a \Longrightarrow v $ (by def. of $\mathbf{R}$ from 2.)

    \item $\Gamma \vdash v : \varphi \in \mathbf{R}$ (by IH from 2. and 3.)

    \item $\Gamma \vdash e' : \EXI\varphi \in \mathbf{R}$ (by def. of $\EXI\varphi \in \mathbf{R}$ from 4.) 
    
\end{enumerate}

\noindent Case $\varphi = \varphi \to \psi$:

\begin{enumerate}

    \item $\Gamma \vdash t : \varphi \to \gamma \in \mathbf{R} $ (by assumption)

    \item $\Gamma \vdash t \Longrightarrow t'$ (by assumption)

    \item $\Gamma \vdash a : \varphi \in \mathbf{R}$ for some arbitrary $a$ (by assumption)

    \item $\Gamma \vdash \textsf{ap}(t, a) : \psi \in \mathbf{R}$  (by def. of $\varphi \to \psi \in \mathbf{R}$ from 1. and 3.) 

    \item $\Gamma \vdash \textsf{ap}( t , a) \Longrightarrow  \textsf{ap}( t' , a)  $ (by \textsf{ap}C from 2. and 4.)

    \item $\Gamma \vdash \textsf{ap}( t' , a)  : \psi \in \mathbf{R}$ (by IH on 4. and 5.)

    \item $\Gamma \vdash t' : \varphi \to \gamma \in \mathbf{R} $ (by def. of $\varphi \to \psi \in \mathbf{R}$ from 6., since it holds for arbitrary $a$)
   
\end{enumerate}

\section{Truncation}
\label{sec:truncation}

Truncation, or more specifically, propositional truncation, is a fundamental concept in many type theories, most notably in homotopy type theory (HoTT), for obtaining controlled propositional proof irrelevance in otherwise proof-relevant frameworks. It is commonly treated as a modality $\TRU$ (see, e.g., \cite{hottbook2013}, \cite{awodey2004}, \cite{kraus2015}, \cite{corfield2020}, \cite{rijke2022}) that allows us to turn an arbitrary type $\varphi$ into a type $\TRU\varphi$ of \textit{mere propositions}, i.e., propositions for which all proofs are considered equal. In other words, truncation hides all the specific information about proofs of a particular proposition except for their existence. Thus, $\TRU\varphi$ represents the proposition that $\varphi$ is just true, i.e., that there exists a proof for it without revealing any more information about it. Or, to use type-theoretic terminology, it represents a proposition that $\varphi$ is inhabited while keeping the inhabitants anonymous.

\subsection{An example} 

How do we use truncation in practice? Let us consider an example originally presented by de Bruijn to highlight the distinctions between proof relevance and proof irrelevance:

\begin{quote}
\small 
the logarithm of a real number is defined for positive numbers only. So actually [from the view of constructive proof-relevant reasoning] the $\textsf{log}$ is a function of \textit{two} variables; and if we use the expression $\textsf{log}(p, q)$, we have to check that $p$ is a real number and $q$ is a proof for the proposition ``$p > 0$''. 
(\cite{debruijn1974}, reprinted in \cite{debruijn1974reprint}, p. 286)
\end{quote}

\noindent So, the type of this common logarithm function $\textsf{log}$ would be:

$$\textsf{log} : (\forall x : \mathbb{R}) \, \textsf{isPositive}(x) \to \mathbb{R}$$

\noindent However, when using the $\textsf{log}$ function, we typically are not concerned with the specific proofs demonstrating the positivity of its arguments, i.e., that $\textsf{isPositive}(x)$ holds of them, we just care about the values of the function. Or to put it differently, we would not want $\textsf{log}$ returning different values for the same argument $p$ just because its proofs of positivity might have been different. The straightforward way to express this with truncation is to simply truncate the corresponding type $\textsf{isPositive}(x)$, and thus obtain a new function $\textsf{log}_\TRU$ with slightly tweaked type:\footnote{This option was, of course, not yet available to de Bruijn in the 70s, as the idea of truncation first appears in the 80s as far as we know (see, e.g., \cite{constable1986}).}

$$\textsf{log}_\TRU : (\forall x : \mathbb{R}) \, \TRU(\textsf{isPositive}(x)) \to \mathbb{R}$$

\noindent Now, what the truncation gives us is that while $\textsf{log}(p, q_1)$ might not be necessarily equal to $\textsf{log}(p, q_2)$ since the proofs $q_1$ and $q_2$ of $\textsf{isPositive}(x)$ might be different, $\textsf{log}_\TRU(p, q_1)$ will be equal $\textsf{log}_\TRU(p, q_1)$ no matter what $q_1$ and $q_2$ we choose as the function $\textsf{log}_\TRU$ does not care about their inner structure only about their existence. In other words, the function $\textsf{log}_\TRU$ would return the same values for the same argument $p$ no matter what its proofs of positivity might look like.

\medskip

It is important to emphasize, however, that truncation is not just a technical tool for simplification as the example above might suggest. It can greatly extend the capabilities of the underlying type-theoretic systems (e.g., it allows us to define subsets and properly formulate the axiom of choice, see, e.g., \cite{hottbook2013}) but, most importantly, it sheds new light on the notions of truth, existence, and proof in a constructive setting. Specifically, it can help us conceptualize the distinction between the existence of proof and the structure of the proof, and thus ultimately its constructive content, a distinction that is often blurred in classical, i.e. non-constructive setting, if present at all. Most importantly, by acknowledging the existence of a proof without committing to any specific construction of it, truncation provides a bridge to classical reasoning. Specifically, with the help of truncation, we can recapture classical reasoning within an otherwise constructive setting. For example, the law of excluded middle can be applied to truncated propositions since they represent just simple truths without additional information, constructive content, or structure. In other words, analogously to a classical proposition, which is either true or false, for a truncated proposition there either exists a proof for it (= it is inhabited) or it does not (= it is empty), there is nothing in between.\footnote{The general need for truncation, or proof irrelevance in general, also indicates that identifying propositions and types, according to the famous dicto ``propositions as types'', might be simply too heavy-handed from some perspectives and that a more careful and discerning approach should be chosen.}

Finally, note that in the above explanation of truncation $\TRU$, the notion of proof existence was repeatedly used. However, this was always done at the metalevel, never directly at the object level. This approach is not unique to our explanation. For example, for Kraus ``$\TRU\varphi$ [in HoTT] intuitively represent[s] the statement that $\varphi$ is inhabited'' (\cite{kraus2015} p. 36) where ``$\varphi$ is inhabited'' is just a different way of saying that ``there exists a proof/construction of $\varphi$''. But HoTT contains no object-level judgment corresponding to the statement ``$\varphi$ is inhabited'', it always appears only as a metalevel informal judgment.\footnote{Later we will see that this causes some inconsistencies with their intended interpretation of the truncation introduction rule (see also Section \ref{sec:related}).}  We only get the propositional approximation of the form $\TRU\varphi$.

\subsection{Truncation and existence} 

Now, let us explore the link between the truncation modality $\TRU$ and the existence modality $\EXI$. In HoTT, propositional truncation is defined via the following three clauses (\cite{hottbook2013}, Sec. 3.7, see also \cite{kraus2015}, Def. 232): 

\begin{enumerate}
\item[(a)] for any $t : \varphi$ we have $ \vert \, t \, \vert : \TRU\varphi$
\item[(b)] if $\psi$ is a mere proposition and we have $f : \varphi \to \psi$, then we have $g : \TRU\varphi \to \psi$
\item[(c)] $\TRU\varphi$ is a mere proposition
\end{enumerate}

\noindent First, note that clause (a) can be regarded as truncation introduction rule and (b) as truncation elimination rule (with added context $\Gamma$):

\begin{center}
\AxiomC{$\Gamma \vdash t : \varphi$}
\UnaryInfC{$\Gamma \vdash \vert \, t \, \vert : \TRU\varphi$}
\DisplayProof
\qquad
\AxiomC{$\Gamma \vdash f : \varphi \to \psi$}
\UnaryInfC{$\Gamma \vdash g : \TRU\varphi \to \psi$}
\DisplayProof
\end{center}

\noindent Now, we will show that our system's direct notational correlates of these rules, that is:\footnote{The symbol ``$\TRU$'' is replaced with ``$\EXI$'' and in (a) ``$\vert \, t \, \vert$'' is replaced with ``$\langle t \rangle$'' and in (b) ``$:$'' is replaced with ``$\existsJ$'' (due to the (b)'s prerequisite that $\psi$ is a mere proposition, we explain this change later when discussing (c)).} 

\begin{center}
\AxiomC{$\Gamma \vdash t : \varphi$}
\UnaryInfC{$\Gamma \vdash \langle t \rangle : \EXI\varphi$}
\DisplayProof
\qquad
\AxiomC{$\Gamma \vdash f \existsJ \varphi \to \psi$}
\UnaryInfC{$\Gamma \vdash g \existsJ \EXI\varphi \to \psi$}
\DisplayProof
\end{center}

\noindent are actually derivable rules. In other words, the I/E rules for truncation modality can be derived from the I/E rules for existence modality, thus showing that truncation $\TRU$ can be reduced to existence $\EXI$. Below we will see that even a judgment corresponding to clause (c) can be established in our system.

Let us start with (a). In our system, we can derive the corresponding rule simply as follows:

\begin{center}
\AxiomC{$\Gamma \vdash t : \varphi $}
\RightLabel{\footnotesize \textsc{just}}
\UnaryInfC{$\Gamma \vdash t \existsJ \varphi $}
\RightLabel{\footnotesize $\EXI$I}
\UnaryInfC{$\Gamma \vdash \langle t \rangle : \EXI\varphi  $}
\DisplayProof
\end{center}

\noindent The rule corresponding to (b) can be derived in the following way:\footnote{Note that the requirement from (b) that $\psi$ has to be a mere proposition corresponds to the fact that $\psi$ has to be a just true proposition as evidenced by the judgment $\ldots \vdash \textsf{ap}(f,x) \existsJ \psi$.}

\begin{prooftree}
\small
    \AxiomC{}
    \RightLabel{\footnotesize \textsc{hyp}}
    \UnaryInfC{$\Gamma , y : \EXI\varphi \vdash y : \EXI\varphi$}

    \AxiomC{$\Gamma \vdash f \existsJ \varphi \to \psi$}
    \RightLabel{\footnotesize \textsc{w}}
    \UnaryInfC{$\Gamma , x : \varphi \vdash f \existsJ \varphi \to \psi$}
    
    \AxiomC{}
    \RightLabel{\footnotesize \textsc{hyp}}
    \UnaryInfC{$\Gamma , x : \varphi  \vdash x : \varphi $}

    \RightLabel{\footnotesize $\to$Ej}
    \BinaryInfC{$\Gamma , x : \varphi  \vdash \textsf{ap}_{\textrm{j}}(f,x) \existsJ \psi $}

    \RightLabel{\footnotesize \textsc{w}}
    \UnaryInfC{$\Gamma , x : \varphi , y : \EXI\varphi \vdash \textsf{ap}_{\textrm{j}}(f,x) \existsJ \psi $}

    \RightLabel{\footnotesize \textsc{ex}}
    \UnaryInfC{$\Gamma , y : \EXI\varphi, x : \varphi  \vdash \textsf{ap}_{\textrm{j}}(f,x) \existsJ \psi $}
    \RightLabel{\footnotesize $\EXI$E}
    \BinaryInfC{$\Gamma , y : \EXI\varphi \vdash \textsf{let } \langle x \rangle \textsf{ be } y \textsf{ in } \textsf{ap}_{\textrm{j}}(f,x) \existsJ \psi $}
    \RightLabel{\footnotesize $\to$Ij}
    \UnaryInfC{$\Gamma \vdash \lambda_{\textrm{j}} y . \textsf{let } \langle x \rangle \textsf{ be } y \textsf{ in } \textsf{ap}_{\textrm{j}}(f,x) \existsJ \EXI\varphi \to \psi $}
\end{prooftree}

\noindent where $g = \lambda_{\textrm{j}} y . \textsf{let } \langle x \rangle \textsf{ be } y \textsf{ in } \textsf{ap}_{\textrm{j}}(f,x)$.

\medskip \noindent \textit{Note.}
Consider the provided informal explanation for clause (a) in \cite{hottbook2013}:

\begin{quote}
\small 
``The first constructor means that if $\varphi$ is inhabited, so is $\TRU\varphi$.'' (\cite{hottbook2013}, p. 117)
\end{quote}

\noindent But note that this is not a precise interpretation: the meaning explanation of the constructor/introduction rule for truncation does not fit the rule itself (observe, however, that it fits exactly our introduction rule for existence $\EXI$I). More specifically, the actual premise of the constructor does not say merely that ``$\varphi$ is inhabited'', i.e., that there exists a proof of it, but there is a stronger claim saying that we actually have a concrete proof for $\varphi$, i.e., ``$t : \varphi$''. In other words, ``$\varphi$ is inhabited'' (= there exists a proof for it) is a different judgment from ``$t : \varphi$'' (= we have a proof $t$ for it).\footnote{Note that this is exactly the ``common mistake which intuitionists are particularly prone to'' observed by \cite{martinlof1993} (p. 141), namely, saying that $\varphi \textit{ exists}$ means that an object of type $\varphi$ has been found. In other words, to borrow Martin-L\"{o}f's terminology, the mistake lies in conflating \textit{bare} existence and \textit{actual} existence.}

Of course, it is easy to diagnose why this imprecise explanation occurs: it is because HoTT has no object-level equivalent of the metalevel judgment ``$\varphi$ is inhabited''. In our system, however, we have an appropriate judgment to capture this: $\varphi \textit{ just true}$ or the typed variant $t \existsJ \varphi$.\footnote{Of course, purely syntactically, the term $t$ is still retrievable from $\langle t \rangle$, however, it is no longer accessible from within the system, i.e., we cannot inspect it by, e.g., running case distinction on it, etc. For all computational/constructive intents and purposes, it is a (non-empty) black box.}

\medskip

Now, let us finally take a closer look at the last clause:

\begin{itemize}
    \item[(c)] $\TRU\varphi$ is a mere proposition.
\end{itemize}

\noindent where a mere proposition is a proposition for which all proofs are equal. As hinted earlier, the requirement that $\TRU\varphi$ is a mere proposition corresponds in our system to the requirement that $\TRU\varphi$  is a just true proposition, i.e., $\Gamma \vdash e \existsJ \TRU\varphi $. The reason for this is that from the perspective of mere propositions, our judgment $e \existsJ \varphi$ (as well as $\varphi \textit{ just true}$) plays double duty. Specifically, it captures two facts about $\varphi$ at the same time: that $\varphi$ is true and that $\varphi$ is also a ``mere proposition'' since it is not only true but specifically \textit{just true} and all \textit{just true} propositions can be shown to be mere propositions. Informally, since $\varphi$ is a just true proposition, it means there exists some proof for it (according to initial definition (2)), let us call it $e$. Now, let us assume there exists another proof for $\varphi$ called $e'$. So, we only know that proofs $e$ and $e'$ of $\varphi$ exist but we have no means of distinguishing between them since we do not know their structure. But having no means of distinguishing between them is just another way of saying that they are effectively equal, which is what we wanted to show.\footnote{As a concrete example of this double duty, in our system we can show that $\Gamma \vdash \varphi \leftrightarrow \EXI\varphi \textit{ just true}$, i.e., that both $\Gamma \vdash \varphi \to \EXI\varphi \textit{ just true}$ and $\Gamma \vdash \EXI\varphi \to \varphi \textit{ just true}$. This essentially corresponds to the HoTT Lemma 3.9.1 (\cite{hottbook2013}, p. 120) stating that if $\varphi$ is a mere proposition, then $\varphi$ and $\TRU\varphi$ are (homotopy) equivalent.} Thus, saying that $\varphi$ is a just true proposition, i.e., that there exists a proof of $\varphi$, formally, $\Gamma \vdash e \existsJ \varphi$, entails that $\varphi$ is a mere proposition. 

It is worth noting that the axiomatic stipulation (c) often comes in a rule form explicitly utilizing (propositional) identity:\footnote{Judgmental identity can be used as well, then leading to judgmental proof irrelevance, but as mentioned earlier, in this paper we focus only on the propositional proof irrelevance.}

\begin{prooftree}
    \AxiomC{$\Gamma \vdash t : \TRU\varphi$}
    \AxiomC{$\Gamma \vdash t' : \TRU\varphi$}
    \BinaryInfC{$\Gamma \vdash t = t' $}
\end{prooftree}

\noindent which states that any two terms $t$ and $ t'$ of type $\TRU\varphi$ are propositionally equal (in other words, this rule effectively stipulates that $\TRU\varphi$ is a proof-irrelevant type).

However, since we do not currently have judgments of the form $\Gamma \vdash e = e' $, i.e., we do not have propositional identity amongst our logical constant, there is no direct correlate in our present system for the rule variant of the stipulation.  For now, we just say that if we were to extend our system with propositional identity (i.e., introduce $e = e'$ as a new type of a proposition together with its corresponding introduction rules, elimination rules, etc.), our correlate of the rule above, i.e., 

\begin{prooftree}
    \AxiomC{$\Gamma \vdash e \existsJ \varphi$}
    \AxiomC{$\Gamma \vdash e' \existsJ \varphi$}
    \BinaryInfC{$\Gamma \vdash e = e' $}
\end{prooftree}

\noindent would not need to be added as an extra stipulation but could be actually derived just as the truncation I/E rules corresponding to clauses (a) and (b) since, as discussed above, any two proofs of \emph{just true} propositions are naturally indistinguishable. Further exploration of the notion of a mere proposition from the perspective of logic of judgemental existence would be interesting but it is beyond the scope of the present paper.

\section{Related work}
\label{sec:related}

\subsection{General remarks}

The investigation of judgmental existence is directly motivated by \cite{martinlof1993} who informally considers a new judgment of the form $\varphi \textit{ exists}$ as expressing the notion of ``bare existence'' (or perhaps more fittingly, ``mere existence''). This notion is explained via the following rule:

\begin{prooftree}
\AxiomC{$a : \varphi$}
\UnaryInfC{$\varphi \textit{ exists}$}
\end{prooftree}

\noindent or, if $\varphi$ is assumed to be a proposition:\footnote{\cite{nordstrom1990} also consider this rule as the ``Proposition as set'' rule: ``If we have an element in a set, then we will interpret that set as a true proposition.'' (p. 37)} 

\begin{prooftree}
\AxiomC{$a : \varphi$}
\UnaryInfC{$\varphi \textit{ true}$}
\end{prooftree}

\noindent That is, $\varphi \textit{ exists }$ has the same meaning as $\varphi \textit{ true}$ if $\varphi$ is a proposition (\cite{martinlof1993}, p. 140). 
To our knowledge, however, he does not ever adopt judgments of the form $\varphi \textit{ exists}$ into a particular type-theoretic system (only in the form of the ``abbreviation'' judgment $\varphi \textit{ true}$ in \cite{martin-lof1984}, which will be discussed below), nor does he consider its internalization into a propositional level via a modality.\footnote{We want to emphasize that our intention is not to reconstruct Martin-L\"{o}f's notion of bare existence or to reduce it into our notion of just truth (i.e., proof existence), we only take it as a source of inspiration.}

In the computational variant of our system, the rule above corresponds to the following rule:

\begin{center}
\AxiomC{$\Gamma \vdash a : \varphi $}
\RightLabel{\footnotesize \textsc{just}}
\UnaryInfC{$\Gamma \vdash a \existsJ \varphi$}
\DisplayProof
\end{center}

\noindent only instead of syntactically hiding the proof expression $a$ we render it ``invisible'' via the proof-irrelevant typing $\existsJ$. That is, instead of $\varphi \textit{ exists}$ we use the more explicit variant $a \existsJ \varphi$. In the logical variant of our system, it would correspond to the rule:

\begin{center}
\AxiomC{$\Gamma \vdash \varphi \textit{ true}$}
\RightLabel{\footnotesize \textsc{just}}
\UnaryInfC{$\Gamma \vdash \varphi \textit{ just true}$}
\DisplayProof
\end{center}

\noindent where $\textit{true}$ captures the proof-relevant notion of truth, and $\textit{just true}$  the proof-irrelevant one (recall Section \ref{sec:intro}).

As mentioned above, in Martin-L\"{o}f's constructive type theory (CTT) (\cite{martin-lof1984}) we can encounter judgments of the form $\varphi \textit{ true}$ which more or less correspond to our judgments of the form $\varphi \textit{ just true}$ from the pure logical variant of our system. However, the judgments of the form  $\varphi \textit{ true}$ from \cite{martin-lof1984} do not seem to be fully incorporated into the system itself. For example, they are missing from the initial specification of the basic judgments of the system and the substitution rules make no use of them. This, of course, makes sense if we regard $\varphi \textit{ true}$ just as an \textit{abbreviation} of the basic judgment $a : \varphi$ which, going by the textual evidence alone, seems to be what Martin-L\"{o}f intented (see \cite{martin-lof1984}, p. 33; see also \cite{nordstrom1990},  p. 25). However, this explanation of the judgment $\varphi \textit{ true}$ is somewhat inconsistent with the one provided in \cite{martinlof1993}. There $\varphi \textit{ true}$ is regarded as having the same meaning as the judgment $\varphi \textit{ exists}$ which is treated as a basic form of a judgment.\footnote{As pointed out by Ansten Klev in personal communication, $\varphi \textit{ exists}$ does not seem to be an abbreviation in any reasonable sense and has to be really understood as a basic form of judgment. We can contrast this, e.g., with the judgment of the form  $\varphi \textit{ false}$, where $\varphi$ is a proposition. Constructively, we can define it as $\varphi \textit{ false} =_{df} \neg \varphi \textit{ true}$ and unlike the meaning explanation of judgment $\varphi \textit{ exists}$ provided in \cite{martinlof1993} this has the form of an explicit definition.} Assuming $\varphi$ is a proposition, this would then mean that all judgments $\varphi \textit{ true}$, $\varphi \textit{ exists}$, and $a : \varphi$ have the same meaning, which cannot be correct.

Note that our proof-relevant judgment $\varphi \textit{ true}$ is essentially just an abbreviation or coding of $a : \varphi$ with no information lost: from $\varphi \textit{ true}$ we should still be able, at least in principle, obtain back the information about the proof $a$ of $\varphi$. This is, however, not the case for proof-irrelevant judgment $\varphi \textit{ just true}$ where the information about $a$ is forgotten and cannot be obtained back. From this perspective, an interesting middle ground was explored by \cite{valentini1998} who proposed a system for ``soft forgetting'': he works with judgments $\varphi \textit{ true}$ which are interpreted similarly to our $\varphi \textit{ just true}$\footnote{``The meaning of $\varphi \textit{ true}$ is that there exists an element $t$ such that $t : \varphi$ but it does not matter which particular element $t$ is'' (\cite{valentini1998}, p. 275) or ``$\varphi$ is inhabited'' (\cite{valentini1998}, p. 278)} but the system he devises allows us to get not only from $a : \varphi$ to $\varphi \textit{ true}$ but also from $\varphi \textit{ true}$ back to $a : \varphi$. 

Our proof-irrelevant judgment $\varphi \textit{ just true}$ is also related to the similar judgment ``$\varphi$ is inhabited'' from Homotopy Type Theory (HoTT, \cite{hottbook2013}). However, we should be careful about conflating these two: in our system, the metalevel judgment ``there exists a proof of $\varphi$'' has a direct object level correlate in the form of the judgment $\varphi \textit{ just true}$, however, in HoTT the metalevel judgment ``$\varphi$ is inhabited'' has no such object level correlate in the form of a judgment.\footnote{From this perspective, it should be evident that our logic of judgmental existence could be also considered as a logic of judgmental inhabitance just by interpreting $\varphi \textit{ just true}$ as ``$\varphi$ is inhabited'' instead of ``there exists a proof of $\varphi$''.}

The notion of existential judgment as such is, of course, much older and can be found in the literature under different names, e.g., ``incomplete communication'' (\textit{unvollst\"{a}ndige Mitteilung}) or ``partial judgment'' (\textit{Partialurteil}) (\cite{hilbert1968}, \cite{kleene1945}) or ``judgment abstract'' (\textit{Urteilsabstrakt}) (\cite{weyl1921}, English translation in \cite{mancosu1998}). In contrast to Weil, however, our existential judgments are considered as proper judgments that describe existential state of affairs which are not regarded just as ``an empty invention of logicians'' (\cite{weyl1921eng}, p. 97).\footnote{Although, as \cite{vandalen1995} points out, Weyl's use of the term ``judgment'' is not in complete agreement with our present-day use of the term.} 

Concerning the relation of our system to other proof-irrelevant ones, \cite{pfenning2001b}, \cite{reed2002}, \cite{abel2012} work with a judgment of the form $a \div \varphi$ that roughly corresponds to our $a \existsJ \varphi$, however, their judgment $a \div \varphi$ itself is not a basic judgment of the system in contrast to our $a \existsJ \varphi$. For them, it is a notational abbreviation for a judgment $a : \varphi$ made under  ``irrelevant'' $\textit{just true}$ assumptions of the form $x \div \varphi$ which, on the other hand, we do not consider. Yet, similar to \cite{pfenning2001b}, we have a two-step truncation:  $\langle a \rangle: \EXI\varphi$ cannot be derived directly from $a : \varphi$, we need the ``pretruncation'' middle step $a \existsJ \varphi$. This is in contrast to \cite{awodey2004}, \cite{hottbook2013}, \cite{kraus2015}. Furthermore, unlike Hofmann's approach (\cite{hofmann1995b}, Proposition 5.3.12), the truncation/existence modality $\EXI$ is derived directly from the judgment level (similarly to \cite{pfenning2001b}) without the need to invoke propositional identity type.

\subsection{Relation to Pfenning and Davies's lax logic}

Our calculus of judgmental existence is directly inspired by \cite{pfenning2001} system of judgmental reconstruction of modal logic,\footnote{The fact that our logic for judgments of the general form ``$\varphi$ is just true'' is similar to the logic for judgments of the general form ``$\varphi$ is possible'' can be made clearer if we consider the following Martin-L\"{o}f's observation: ``The important point to observe here is the change from \textit{is} in $A$ is true to \textit{can} in $A$ can be verified, or $A$ is verifiable. Thus what is expressed in terms of being in the first formulation really has the modal character of possibility'' (\cite{martinlof1996}, p. 25).} which is in turn inspired by Martin-L\"{o}f's general methodology of distinguishing between judgments and propositions. Formally, our calculus shares the most resemblance with their possibility logic (although our system does not contain necessity modality and no corresponding validity assumptions of the form $\varphi \textit{ valid}$), and especially their lax logic: our I/E rules for existence modality $\EXI$ directly correspond to their I/E rules for lax modality $\medcircle$. Actually, we can show that our $\EXI$ modality obeys the lax modality $\medcircle$ axioms (\cite{fairtlough1997}):

\begin{itemize}
    \item[i)] $\varphi \to \medcircle \varphi$
    \item[ii)] $\medcircle  \medcircle  \varphi \to \medcircle \varphi$
    \item[iii)] $ (\varphi \to \psi) \to  \medcircle \varphi \to  \medcircle \psi $
\end{itemize}

\begin{proposition} \phantom{}

\begin{enumerate}
    \item[i)] $\Gamma \vdash \varphi \to \EXI \varphi \textit{ true}$
    \item[ii)] $\Gamma \vdash \EXI \EXI \varphi \to \EXI\varphi \textit{ true}$
    \item[iii)] $\Gamma \vdash  (\varphi \to \psi ) \to  \EXI\varphi  \to \EXI\psi \textit{ true}$
\end{enumerate}    

\end{proposition}

\noindent \textit{Proof}.

\medskip

i)

\begin{prooftree}
    \AxiomC{}
    \RightLabel{\footnotesize \textsc{hyp}}
    \UnaryInfC{$\Gamma , \varphi \textit{ true} \vdash \varphi \textit{ true}$}
    \RightLabel{\footnotesize \textsc{just}}
    \UnaryInfC{$\Gamma , \varphi \textit{ true} \vdash \varphi \textit{ just true}$}
    \RightLabel{\footnotesize $\EXI$I}
    \UnaryInfC{$\Gamma , \varphi \textit{ true} \vdash \EXI \varphi \textit{ true}$}
    \RightLabel{\footnotesize $\to$I}
    \UnaryInfC{$\Gamma \vdash \varphi \to \EXI \varphi \textit{ true}$}
\end{prooftree}

\medskip

ii) See Proposition \ref{prop:6}.

\medskip

iii) 

\begin{prooftree}
\small
\AxiomC{}
\RightLabel{\footnotesize \textsc{hyp}}
\UnaryInfC{$\Gamma , \EXI\varphi \vdash \EXI\varphi $}
\RightLabel{\footnotesize \textsc{w}}
\UnaryInfC{$\Gamma , \EXI\varphi , \varphi \to \psi \vdash \EXI\varphi $}

\AxiomC{}
\RightLabel{\footnotesize \textsc{hyp}}
\UnaryInfC{$\Gamma , \varphi \to \psi \vdash \varphi \to \psi $}

\RightLabel{\footnotesize \textsc{w}}
\UnaryInfC{$\Gamma , \varphi \to \psi , \varphi \vdash \varphi \to \psi$}

\AxiomC{}
\RightLabel{\footnotesize \textsc{hyp}}
\UnaryInfC{$\Gamma , \varphi \vdash \varphi $}
\RightLabel{\footnotesize \textsc{w}}
\UnaryInfC{$\Gamma , \varphi , \varphi \to \psi \vdash \varphi $}
\RightLabel{\footnotesize \textsc{ex}}
\UnaryInfC{$\Gamma ,  \varphi \to \psi , \varphi \vdash \varphi $}

\RightLabel{\footnotesize $\to$E}
\BinaryInfC{$\Gamma ,  \varphi \to \psi , \varphi \vdash \psi $}

\RightLabel{\footnotesize \textsc{just}}
\UnaryInfC{$\Gamma ,  \varphi \to \psi , \varphi  \vdash \psi \textit{ just}$}

\RightLabel{\footnotesize \textsc{w}}
\UnaryInfC{$\Gamma ,  \varphi \to \psi , \varphi , \EXI\varphi  \vdash \psi \textit{ just}$}

\RightLabel{\footnotesize \textsc{ex}}
\UnaryInfC{$\Gamma ,  \EXI\varphi , \varphi \to \psi , \varphi \vdash \psi \textit{ just}$}

\RightLabel{\footnotesize $\EXI$E}
\BinaryInfC{$\Gamma , \EXI\varphi , \varphi \to \psi \vdash \psi \textit{ just}$}

\RightLabel{\footnotesize $\EXI$I}
\UnaryInfC{$\Gamma , \EXI\varphi , \varphi \to \psi \vdash  \EXI\psi $}

\RightLabel{\footnotesize \textsc{ex}}
\UnaryInfC{$\Gamma , \varphi \to \psi , \EXI\varphi \vdash  \EXI\psi $}

\RightLabel{\footnotesize $\to$I}
\UnaryInfC{$\Gamma , \varphi \to \psi \vdash  \EXI\varphi  \to \EXI\psi $}

\RightLabel{\footnotesize $\to$I}
\UnaryInfC{$\Gamma \vdash  (\varphi \to \psi ) \to  \EXI\varphi  \to \EXI\psi $}

\end{prooftree}

\noindent Thus, we can conclude that $\EXI$ is a lax modality.

However, at first glance, there appears to be a clear difference between Pfenning and Davies's system of lax logic and ours: their system does not utilize the just true variants (or rather in their case, the lax variants) of I/E rules (see Section \ref{sec:formal_comp}).

First, let us show a fragment of a logical variant of our system without these just true variants of I/E rules (i.e., without the j-variants $\to$Ij, $\to$Ej, $\EXI$Ij, $\EXI$Ej), is actually equivalent to Pfenning and Davies's system of lax logic. First, let us specify a translation function $\ell$ for propositions and contexts of our system into Pfenning and Davies's lax logic (\cite{pfenning2001}): $ \ell ( \varphi \to \psi )  = \ell (  \varphi ) \Rightarrow \ell ( \psi ) $, $ \ell ( \EXI\varphi ) = \medcircle \ell ( \varphi ) $, $  \ell (  p )  = p $, $  \ell ( \Gamma , \varphi \textit{ true}  ) =  \ell ( \Gamma ) ,  \ell ( \varphi ) \textit{ true}$.

Let us denote judgments derivable in our (logical) fragment as $\Gamma \vdash_{E} \varphi \textit{ true}$ and judgments derivable in Pfenning and Davies's lax logic as $\Gamma \vdash_L \varphi \textit{ true}$.

\begin{proposition} \label{prop:lax_correctness} (Correctness of translation)
 \begin{enumerate}

    \item[] If $\Gamma \vdash_{E} \varphi \textit{ true}$ then $\ell ( \Gamma ) \vdash_L \ell (  \varphi ) \textit{ true}$.
    \item[] If $\Gamma \vdash_{E} \varphi \textit{ just true}$ then $\ell ( \Gamma ) \vdash_L \ell ( \varphi ) \textit{ lax}$.

\end{enumerate}   
\end{proposition}

\begin{proof}
By simultaneous induction on the structure of the given derivations. This is straightforward, as the logical rules of our fragment and of Pfenning and Davies's lax logic closely correspond to each other (we essentially just need to replace $\EXI$ with $\medcircle$ and $\textit{just true}$ with $\textit{lax}$).   
\end{proof}

Now, let us specify a translation function $\varepsilon$ for propositions and contexts of Pfenning and Davies's lax logic into our system: $ \varepsilon( \varphi \Rightarrow \psi ) = \varepsilon( \varphi ) \to \varepsilon( \psi ) $, $\varepsilon( \medcircle\varphi ) = \EXI \varepsilon( \varphi )$, $\varepsilon( p ) = p $, $ \varepsilon( \Gamma , \varphi \textit{ true} ) =  \varepsilon( \Gamma ) , \varepsilon( \varphi ) \textit{ true}$. It holds that $ \varepsilon (  \ell ( \varphi ) ) = \varphi$.

\begin{proposition} (Completness of translation) \label{prop:lax_completness}
\begin{enumerate}
    \item[] If $\Gamma \vdash_L \varphi \textit{ true}$ then $\varepsilon ( \Gamma ) \vdash_{E} \varepsilon ( \varphi ) \textit{ true}$.
    \item[] If $\Gamma \vdash_L \varphi \textit{ lax}$ then $\varepsilon ( \Gamma ) \vdash_{E} \varepsilon ( \varphi ) \textit{ just true}$.
\end{enumerate}    
\end{proposition}

\begin{proof}
By simultaneous induction on the structure of the given derivations.

Putting together Propositions \ref{prop:lax_correctness} and \ref{prop:lax_completness}, from  
$ \ell(\Gamma) \vdash_L  \ell (\varphi)\textit{ true}$ we can derive $\varepsilon ( \ell( \Gamma  )) \vdash_E \varepsilon ( \ell( \varphi  ))  \textit{ true}$, and thus $\Gamma \vdash_{E} \varphi \textit{ true}$, and finally conclude that $\Gamma \vdash_{E} \varphi \textit{ true}$ iff $\ell(\Gamma) \vdash_L  \ell (\varphi)\textit{ true}$.

\end{proof}

Thus, a fragment of our logical system without the rules $\to$Ij, $\to$Ej, $\EXI$Ij, $\EXI$Ej can be translated into Pfenning and Davies's lax logic. However, we can go a little bit further. As it turns out, all the \textit{just true} variants of rules $\to$Ij, $\to$Ej, $\EXI$Ij, $\EXI$Ej can be actually derived in our system from the standard \textit{true} variants $\to$I, $\to$E, $\EXI$I, $\EXI$E if we add the following rule (note that it corresponds to the right-to-left direction of Proposition \ref{prop:1}): 

\begin{prooftree}
    \AxiomC{$\Gamma \vdash  \varphi \to  \EXI \psi \textit{ true}$}
    \RightLabel{\footnotesize \textsc{r}}
    \UnaryInfC{$\Gamma \vdash  \varphi \to  \psi \textit{ just true}$}
\end{prooftree}

\noindent This also means that we can reduce a full logical variant of our system into Pfenning and Davies's lax logic extended by the rule \textsc{r}. 

\begin{proposition} \label{prop:lax_derivability}
The rules $\to$Ij, $\to$Ej, $\EXI$Ij, $\EXI$Ej are derivable from $\to$I, $\to$E, $\EXI$I, $\EXI$E together with the structural rules and the \textsc{r} rule.    
\end{proposition}

\noindent \textit{Proof}.

\medskip

$\to$Ij: 

\begin{prooftree}
    \AxiomC{$\Gamma , \varphi \textit{ true} \vdash \psi \textit{ just true}$}
    \RightLabel{\footnotesize $\EXI$I}
    \UnaryInfC{$\Gamma , \varphi \textit{ true} \vdash \EXI \psi \textit{ true}$}
    \RightLabel{\footnotesize $\to$I}
    \UnaryInfC{$\Gamma \vdash  \varphi \to  \EXI \psi \textit{ true}$}
    \RightLabel{\footnotesize \textsc{r}}
    \UnaryInfC{$\Gamma \vdash  \varphi \to  \psi \textit{ just true}$}

\end{prooftree}

$\to$Ej:

\begin{prooftree}
    \AxiomC{$\Gamma \vdash \varphi \to \psi \textit{ just}$}
    \RightLabel{\footnotesize $\EXI$I}
    \UnaryInfC{$\Gamma \vdash \EXI (\varphi \to \psi)$}

    \AxiomC{}
    \RightLabel{\footnotesize \textsc{hyp}}
    \UnaryInfC{$\Gamma , \varphi \to \psi \vdash \varphi \to \psi $}

    \AxiomC{$\Gamma \vdash \varphi$}
    \RightLabel{\footnotesize \textsc{w}}
    \UnaryInfC{$\Gamma , \varphi \to \psi  \vdash \varphi$}
    \BinaryInfC{$\Gamma , \varphi \to \psi  \vdash \psi $}
    \RightLabel{\footnotesize \textsc{just}}
    \UnaryInfC{$\Gamma , \varphi \to \psi \vdash \psi \textit{ just}$}

    \RightLabel{\footnotesize $\EXI$E}
    \BinaryInfC{$\Gamma \vdash  \psi \textit{ just}$}
\end{prooftree}

$\EXI$Ij:

\begin{prooftree}
    \AxiomC{$\Gamma \vdash \varphi \textit{ just true}$}
    \RightLabel{\footnotesize $\EXI$I}
    \UnaryInfC{$\Gamma \vdash \EXI\varphi \textit{ true}  $}
    \RightLabel{\footnotesize \textsc{just}}
    \UnaryInfC{$\Gamma \vdash \EXI\varphi \textit{ just true}  $}
\end{prooftree}

$\EXI$Ej:

\begin{prooftree}
    \AxiomC{$\Gamma \vdash \EXI \varphi \textit{ just}$}
    \RightLabel{\footnotesize $\EXI$I}
    \UnaryInfC{$\Gamma \vdash \EXI  \EXI \varphi $}

    \AxiomC{}
    \RightLabel{\footnotesize \textsc{hyp}}
    \UnaryInfC{$\Gamma , \EXI \varphi   \vdash \EXI \varphi  $}
    
    \AxiomC{$\Gamma , \varphi \vdash \gamma \textit{ just}$}

    \RightLabel{\footnotesize \textsc{w}}
    \UnaryInfC{$\Gamma , \varphi , \EXI \varphi \vdash \gamma \textit{ just}$}
    \RightLabel{\footnotesize \textsc{ex}}
    \UnaryInfC{$\Gamma , \EXI \varphi , \varphi  \vdash \gamma \textit{ just}$}
    \RightLabel{\footnotesize $\EXI$E}
    
    \BinaryInfC{$\Gamma , \EXI \varphi \vdash \gamma \textit{ just} $}
    \RightLabel{\footnotesize $\EXI$E}
    \BinaryInfC{$\Gamma \vdash \gamma   \textit{ just} $}
    
\end{prooftree}

\section{Concluding remarks}

The main novelty of our system in comparison to others found in the literature is that it allows a rigorous treatment of basic judgments of the form ``there exists a proof of $\varphi$'' which are in other systems (e.g., HoTT, CTT) treated only as metalevel judgments or internalized as modal propositions. Furthermore, our system distinguishes between \textit{judgmental} and \textit{propositional existence}, while the latter is derived from the former. The link between propositional existence $\EXI$ and truncation $\TRU$ thus suggests that the notion of truncation is ultimately based on the notion of judgmental existence. From this perspective, it would be interesting to explore a weaker variant of truncation introduced by \cite{kraus2015} and called anonymous existence.\footnote{An unexpected connection between truncation and the notion of presupposition of inquisitive logic was recently observed by \cite{puncocharpezlar2024} (forthcoming), which also suggests a link to our calculus of judgmental existence.} Similarly, the links between the existence modality $\EXI$ and lax modality $\medcircle$ of \cite{fairtlough1997} and intuitionistic knowledge modality $\textsf{K}$ of \cite{artemov2016}\footnote{We thank Szymon Chlebowski for bringing our attention to this paper.} seem intriguing and would deserve further investigation. It also remains an open question what exactly is the relationship between the proof-relevant judgments $\varphi \textit{ true}$ and proof-irrelevant existential judgments $\varphi \textit{ just true}$. We have mentioned that the former entails the latter but not vice versa, however, more can be definitely said about this relationship.

The main limitation of the present calculus is that we are considering only a propositional fragment with $\to$ and $\EXI$. Extending the system with other connectives (including identity), then generalizing it towards first-order logic, and finally towards dependent type theory are the most straightforward directions for future development.\footnote{Or, alternatively, extending systems like HoTT or CTT with judgments of the form $\varphi \textit{ inhab}$/$a \existsJ \varphi$. It is also worth noting that in impredicative versions of HoTT, it is possible to define truncation $\TRU$ without introducing a new type constructor. This suggests that existence $\EXI$ might also be definable in such a way. Furthermore, the ability to define truncation as a higher-inductive type opens the question whether existence $\EXI$ could be defined in a similar manner.} The addition of propositional identity (so that we could form new propositions such as $p : a = b$ and thus also $p \existsJ a = b$) would be especially interesting as it would allow us to fully explore the notion of propositional proof irrelevance and related notions such as mere propositions. A natural next step would be to also adopt judgmental identity (so that we could form new judgments such as $a \equiv b : \varphi$ and thus also $a \equiv b \existsJ \varphi$) as it would allow us to start exploring the notion of judgmental proof irrelevance as well.


{
}

\end{document}